\documentclass[pdftex,11pt,a4paper] {article}
\usepackage{amsmath,amsfonts,amsthm,amssymb,amscd}
\usepackage{latexsym}
\usepackage[pdftex]{graphicx}
\usepackage{float}
\usepackage{xcolor}
\usepackage{multirow}
\usepackage{enumitem}   
\usepackage{bbm}
\usepackage{natbib} 

\usepackage{caption}
\usepackage{geometry}
\newtheorem{corollary}{Corollary}[section]
\newtheorem{lemma}{Lemma}[section]

\newtheorem{proposition}{Proposition}[section]
\newtheorem{definition}{Definition}[section]

\def\bitcoinB{\leavevmode
  {\setbox0=\hbox{\textsf{B}}%
    \dimen0\ht0 \advance\dimen0 0.2ex
    \ooalign{\hfil \box0\hfil\cr
      \hfil\vrule height \dimen0 depth.2ex\hfil\cr
    }%
  }%
}

\usepackage[nodayofweek]{datetime}
\newdate{date}{28}{10}{2026}
\date{\displaydate{date}}

\usepackage{hyperref}  
\hypersetup{colorlinks=true,
    linktoc=all,
    linkcolor=blue,
		citecolor=blue,
		urlcolor=cyan,
		bookmarksopen=true
		}

\begin{document}

\title{Unified Approach for Hedging Impermanent Loss of Liquidity Provision}

\author{Alex Lipton\thanks{Abu Dhabi Investment Authority \& ADIA Lab \& Khalifa University, Abu Dhabi, UAE}, Vladimir Lucic\thanks{Imperial College \& Marex, London, UK} and Artur Sepp\thanks{Corresponding Author, LGT Bank, Zurich, Switzerland, artursepp@gmail.com}}

\maketitle

\begin{abstract}
We develop static model-independent and dynamic model-dependent approaches for hedging of the impermanent loss (IL) of liquidity provision (LP) staked at Decentralised Exchanges (DEXes) which employ Uniswap V2 and V3 protocols. We provide detailed definitions and formulas for computing the IL to unify the different definitions occurring in the existing literature. We show that the IL can be seen as a contingent claim with a non-linear pay-off for a fixed maturity date. Thus, we introduce the contingent claim termed the IL protection claim which delivers the negative of IL pay-off at the maturity date. We apply arbitrage-based methods for the valuation and risk management of this claim. First, we develop the static model-independent replication method for the valuation of IL protection claim using traded European vanilla call and put options. We extend and generalise an existing method to show that the IL protection claim can be hedged perfectly with options if there is a liquid options market. Second, we develop the dynamic model-based approach for the valuation and hedging of IL protection claims under a risk-neutral measure. We derive analytic valuation formulas using a wide class of price dynamics for which the characteristic function is available under the risk-neutral measure. As base cases, we derive analytic valuation formulas for the IL protection claim under the Black-Scholes-Merton model and the log-normal stochastic volatility model. We finally discuss the estimation of the risk-reward of LP staking using our results.
\end{abstract}

\textit{Keywords}: Automated Market Making, Liquidity Provision, Decentralised Finance, Uniswap, Cryptocurrencies, Impermanent Loss

\textit{JEL Classifications}:  C02, G12, G23

\section{Introduction}

Decentralised Exchanges (DEXes) play a fundamental role in the blockchain
ecosystem by allowing users to swap digital assets. The functioning of DEXes
requires liquidity providers to stake their liquidity in so-called pools,
so that traders can use these pools to buy and sell tokens. Automated
Market Making (AMM) protocol is a mechanism to settle buy and sell orders
on DEXes. An AMM protocol is characterised by a constant function market
maker (CFMM) which assigns buy and sell prices for given orders using order
sizes and current liquidity of a pool. Uniswap V3 (\citet{V3}) is the most
widely employed CFMM which is adopted by many DEXes, in addition to Uniswap
DEX itself. Uniswap V2 (\citet{V2}) is an earlier AMM protocol that employs the constant product CFMM
that is less capital efficient than the V3 CFMM. Uniswap V2 is still in use for
old altcoin pools. For details of the design of various CFMMs, see, among others, \citet{Angeris2019}, \citet{LiptonTreccani2021}, \citet{Mohan2022}, %
\citet{LiptonHardjono2022}, \citet{LiptonSepp2022}, \citet{Milionis2022}, %
\citet{Park2023}, \citet{Lehar2024}. We note that in early 2025, the Uniswap development team launched the Uniswap V4 (\citet{V4}) AMM, which has the same CFMM as Uniswap V3, so our results will be directly applicable to Uniswap V4 pools. \citet{Alexander2024} establish empirically that technological advances in AMM protocols, such as improvements in Uniswap V3 compared to V2, lead to significant improvements in the liquidity and price discovery of DEXes. In this respect, our results may help liquidity providers manage risk of their liquidity provisions and thus enhance the liquidity of DEXes.

Uniswap V3 protocol allows liquidity providers to concentrate liquidity in
specified ranges. As a result, the liquidity of the pool can be increased in
certain ranges (typically around the current price), and the potential to
generate more trading fees from the LP is increased accordingly. We
illustrate the dynamics of staked LP using the ETH/USDT pool as an example. A
liquidity provider stakes liquidity in a specific range using an initial amount
of ETH and USDT tokens as specified by Uniswap V3 CFMM. When the price of
ETH falls, traders use the pool to swap USDT by depositing ETH, so that the
LP accrues more units of ETH. Thus, when ETH falls persistently, the
liquidity provider ends up holding more units of the depreciating asset,
which is similar to being short a put option. In contrast, when the price of ETH
increases, traders will deplete the reserves of ETH from the pool by depositing
USDT tokens. Thus, the liquidity provider ends up holding less units of the
appreciating asset, which is similar to being short a call option. The
combined effect of increasing / decreasing the exposure to depreciating /
appreciating assets leads to what is known as the impermanent loss in
Decentralised Finance (DeFi) applications.

It is clear that the mechanism behind the impermanent loss of LPs is similar
to being short a portfolio of call and put options. Consequently, we can
apply well-developed methods from financial engineering for designing the
valuation and risk management of LPs using static and dynamic methods for
the valuation and risk management of derivative securities. Our contributions
include static model-independent replication and dynamic model-dependent
replication of IL.

\subsection{Literature review and contributions}

For the static replication of IL under Uniswap V2 protocol, %
\citet{Fukasawa2023} derive approximate hedging portfolio using variance and
gamma swaps. In addition, \citet{Lipton2024} extends \citet{LiptonSepp2008}, 
\citet{Lipton2018}, and obtains model-dependent costs
of hedging portfolios using values of variance and gamma swaps under Heston
model. \citet{Deng2023} and \citet{MaireWunsch2024} develop the static
replication for the Uniswap V3 protocol. We note that their result hinges
specifically on the analytical formula of the IL under Uniswap V3 AMM. Also
this method requires holding both call and put options for in-the-money and
out-of-the-money strikes, which could be very costly for practical
implementation\footnote{%
At major options exchanges for digital and traditional markets, the
liquidity is concentrated in call and put options with out-of-the-money and
near at-the-money strikes. Options with in-the-money strikes are not liquid
with high bid-ask spreads. As a result, a liquid and cost-efficient
replicating portfolio should include only out-of-the-money and near-at-the-money strikes}. We provide a generic approach to designing a cost-effective replicating portfolio for IL hedging for general CFMM and
illustrate our approach for the Uniswap V3 protocol.

We note that currently the liquid option market (on Deribit options
exchanges and some other centralised exchanges) exists only for Bitcoin and
Ether. Thus, hedging of LP stakes for altcoins requires a model-dependent
approach for replication of IL by introducing the IL protection claim and
applying dynamic delta-hedging of this claim. In the current ecosystem, some market-making and trading companies provide off-chain (over-the-counter)
IL protection claims for a wide range of digital assets. We develop a
dynamic model-dependent approach for the valuation of IL protection claims,
which appears to be new in the literature.

For valuation purposes under the Uniswap V3 protocol, we provide an original
result with a decomposition of the IL into components including pay-offs of
vanilla call and put options, digital options, and an exotic pay-off on the
square root of the price. This decomposition allows for the valuation of the
IL protection claims under a large class of price dynamics, which have
solution for their characteristic functions, by utilising the Lipton-Lewis
formula. As an important example, we derive an analytic solution for IL
protection claim under Black-Scholes-Merton model, which allows us to analyse
the value of the IL protection claim using a single volatility parameter. As
realistic price dynamics that include stochastic volatility correlated with the
price dynamics, we apply the log-normal stochastic volatility model
developed in \citet{SeppRakhmonoV2023}.

Empirical analysis of the profitability of LP strategies is currently a very active area
in the AMM-related literature. We refer, among others, to
\citet{Heimbach2022}, \citet{Cartea2023}, \citet{Bergault2023}, %
\citet{Cartea2024}, \citet{Li2023}, \citet{Li2024}, \citet{Bronnimann2024}. In particular, \citet{Cartea2024} consider a diffusion model with a stochastic mean and constant volatility for the dynamics of the swap rate and derive an optimal liquidity range assuming logarithmic utility of LP providers. Their results depend on the constant volatility of the swap dynamics. In such a setup, our results allow for mitigation of model risk by either using a static replication when liquid options are available for trading (in which case, the model risk is eliminated) or using realistic dynamics for the swap rate including jumps and stochastic volatility with dynamic replication of IL protection claim.

Our paper is organised as follows. In Section \ref{sec:lp}, we provide
definitions and derivations of IL and pay-offs of IL protection claims.
In Section \ref{sec:apps}, we apply these results for the Uniswap V2 and V3
protocols. Hereby, we derive the decomposition formula for IL in Uniswap V3
protocol into pay-offs of vanilla, digital, and square-root contracts, which
we use further for model-dependent valuation. In Section \ref{sec:sh}, we
develop a generic approach for static replication of IL using traded vanilla
options. In Section \ref{sec:mdp}, we develop the model-dependent approach
for the valuation of protection claim against IL. We conclude in Section \ref{sec:conc}.

\section{Impermanent Loss of Liquidity Provision}

\label{sec:lp}

\subsection{Liquidity Provision}

We consider a liquidity pool on a pair of token 1 and token 2. Without loss
of generality, we assume that token 1 is a volatile token and token 2 is a
stable token with the spot price $p$ of swapping one unit of token 1 to $p$
units of token 2. For concreteness, we fix token 1 to be ETH and token 2 to
be USDT with spot price $p$ being ETH/USDT exchange price ($p_{0}=1600$ as
of 19 April 2025).

We consider a liquidity provision (LP) provided for $x_{0}$ and $y_{0}$ units
of token 1 (ETH) and token 2 (USDT), respectively. The initial value of the
LP in the units of token 2 (USDT) is given by\footnote{We note that throughout the paper, the value of LP $V$ and its allocated units $x$ and $y$ characterize the state of a stand-alone LP and not the state of the aggregate pool.}
\begin{equation}  \label{eq:pnl1}
V^{(y)}_{0} = p_{0} x_{0} + y_{0}.
\end{equation}

The value of the LP at time $t$ is given by 
\begin{equation}  \label{eq:pnl2}
V^{(y)}_{t} = p_{t} x_{t} + y_{t},
\end{equation}
where $x_{t}$ and $y_{t}$ are the current units of ETH and USDT (these units
are the outputs from AMM protocol), respectively, in the staked LP and $%
p_{t} $ is the current ETH/USDT spot price. We treat accrued LP fees
separately in line with the Uniswap convention.

The value of the LP in units of token 1 is given by 
\begin{equation}  \label{eq:pnl1_}
V^{(x)}_{t} = x_{t} + p^{-1}_{t} y_{t}.
\end{equation}
Our further results can be directly applied for pools with USDT/ETH type of
conversion using corresponding inverse prices and ranges. We also note that
in Uniswap V2 and V3, the price is defined on the grid of price ticks which
are functions of the pool fee tiers. Price ticks are dense for pools with
small fee tiers, so we assume that price range for $p_{t}$ is continuous
(see \cite{Echenim2023} for the analysis using discrete ticks).

\subsection{Profit-and-Loss of LP}

We consider two types of LP strategies that exclude and include static delta
hedging of the initial stake.

\begin{definition}[USDT Funded LP position]
Funded position is created by funding the initial allocation of $x_{0}$ and $%
y_{0}$ units with the total capital commitment of $V^{(y)}_{0}$ USDT.
The stake in the funded position includes the purchase of $x_{0}$ units of ETH
token at the price $p_{0}$.
\end{definition}

The value of the funded position equals the value of LP in Eq. (\ref{eq:pnl2}): $V \ funded^{(y)}_{t} = V^{(y)}_{t}$. As a result, the
Profit-and-Loss (P\&L) of the funded position in token 2 (USDT) at time $t$
is given by 
\begin{equation}  \label{eq:pnl_funded}
\text{P\&L funded}^{(y)} \equiv V^{(y)}_{t} - V^{(y)}_{0} = \left(p_{t} x_{t} +
y_{t}\right) - \left(p_{0}x_{0} + y_{0}\right).
\end{equation}

\begin{definition}[Borrowed LP position]
Borrowed position is created either by borrowing $x_{0}$ units of ETH or by
purchasing $x_{0}$ units of ETH for staking and by simultaneously selling
short the perpetual future for hedging the initial stake of $x_{0}$ units of
ETH\footnote{%
There is very liquid market for core cryptocurrencies on both on-chain
exchanges (such as Hyperliquid, GMX, Aevo) and on off-chain exchanges (such
as Binance, Bybit, Deribit), so that hedging of long exposures is possible.}.
\end{definition}

For the borrowed position with hedging, we set the hedge position of selling
short $x_{0}$ units of token 1 with strike/entry price $p_{0}$. The P\&L of
the hedge position in units of USDT token at time $t$ is given by 
\begin{equation}  \label{eq:pnl4}
\text{Hedge}^{(y)}_{t}=-\left( p_{t} - p_{0}\right) x_{0}.
\end{equation}
We assume that the hedge can be implemented by short-selling the perpetual
future and we treat the funding cost separately from the LP P\&L. The
initial value of the staking position is given in Eq. (\ref{eq:pnl2}). The
value of the borrowed LP position at time $t$ is given by 
\begin{equation}  \label{eq:v_borrowed}
\text{V borrowed}^{(y)}_{t} = p_{t} x_{t} + y_{t} -\left[p_{t}-p_{0}\right]x_{0}.
\end{equation}
The P\&L of the borrowed position at time $t$ is given by 
\begin{equation}  \label{eq:pnl_borrowed}
\begin{split}
\text{P\&L borrowed}^{(y)}_{t} & = \left(p_{t} x_{t} + y_{t} -\left[p_{t}-p_{0}%
\right]x_{0}\right) - \left(p_{0}x_{0} + y_{0}\right) \\
& = \left(p_{t} x_{t} + y_{t} \right) - \left(p_{t} x_{0} + y_{0}\right).
\end{split}%
\end{equation}

\subsection{Impermanent Loss}

Using the two definitions of LPs in Eqs \eqref{eq:pnl_funded} and \eqref{eq:pnl_borrowed}, we define the quantity known as impermanent
loss in the following three ways.

\begin{definition}[Impermanent Loss $\%$]
\label{def:il} The nominal IL of funded LP is defined for the LP funded with
USDT by 
\begin{equation}  \label{eq:il_funded}
\begin{split}
& \mathrm{IL \ funded}^{(y)}(p_{t}) = \frac{\left(p_{t} x_{t} +
y_{t}\right)-\left(p_{0}x_{0} + y_{0}\right)}{\left(p_{0}x_{0} + y_{0}\right)%
}.
\end{split}%
\end{equation}
The nominal IL of borrowed LP is defined for the LP with borrowed ETH by 
\begin{equation}  \label{eq:il_borrowed_nom}
\begin{split}
& \mathrm{IL \ borrowed}^{(y)}(p_{t}) = \frac{\left(p_{t} x_{t} + y_{t} \right) -
\left(p_{t} x_{0} + y_{0}\right)}{\left(p_{0}x_{0} + y_{0}\right)}.
\end{split}%
\end{equation}
The relative IL of borrowed LP is defined for borrowed LP by 
\begin{equation}  \label{eq:il_borrowed_rel}
\begin{split}
& \mathrm{Rel \ IL \ borrowed}^{(y)}(p_{t}) = \frac{\left(p_{t} x_{t} + y_{t} \right)
- \left(p_{t} x_{0} + y_{0}\right)}{\left(p_{t}x_{0} + y_{0}\right)}.
\end{split}%
\end{equation}
\end{definition}

In the literature, all three definitions are being used. Here, we clarify
the meaning of each definition. The nominal IL of the funded LP is applicable
when the LP provider funds the position by allocation $V_{0}$ USDT tokens
and buys the initial stake of $x_{0}$ ETH tokens. The nominal IL for
borrowed LP is common for LPs accompanied either by borrowing the initial
stake of $x_{0}$ ETH tokens or by employing static delta-hedging of the initial
stake of $x_{0}$ tokens using perpetual futures. The relative IL of borrowed
LP defines the P\&L of the borrowed LP relative to the buy-and-hold position
rather than the initial staked value of the LP.

The nominal IL can be easily interpreted because P\&L of the LP in USDT is
the nominal IL multiplied by the initial staked notional so that we obtain 
\begin{equation}  \label{eq:V3aaa}
\begin{split}
\text{P\&L funded}^{(y)}(p_{t}) & = N^{(y)} \times \ \text{IL funded}^{(y)}(p_{t}) \\
\text{P\&L borrowed}^{(y)}(p_{t}) & = N^{(y)} \times \ \text{IL borrowed}^{(y)}(p_{t}),
\\
\end{split}%
\end{equation}
where $N^{(y)}$ is the initial notional in USDT token of the staked LP. In
contrast, the relative IL lacks this interpretation. Thus, while relative IL
appears in some of the literature to emphasize the IL relative to the
buy-and-hold portfolio, its practical application for modelling of the
realised P\&L from an LP is not obvious. In this paper, we focus only on
hedging of the nominal IL for funded and borrowed LPs in Eqs (\ref{eq:il_funded}) and (\ref{eq:il_borrowed_nom}), respectively.

Since the tokens can be used interchangeably, our definitions are symmetric.
For a position funded in $x$ (ETH) tokens the corresponding P\&L is obtained
using $p^{-1}_{t}=1/p_{t}$ and Eq. \eqref{eq:V3aaa} becomes 
\begin{equation}  \label{eq:V3x}
\begin{split}
\text{P\&L funded}^{(x)}(p^{-1}_{t}) & = N^{(x)} \times \text{IL funded}^{(x)}(p^{-1}_{t}), \\
\text{P\&L borrowed}^{(x)}(p^{-1}_{t}) & = N^{(x)} \times \text{IL borrowed}^{(x)}(p^{-1}_{t}),
\end{split}%
\end{equation}
where $N^{(x)}$ is notional in $x$ units and 
\begin{equation}  \label{eq:il_funded_x}
\begin{split}
&\text{IL funded}^{(x)}(p^{-1}_{t}) = \frac{\left(x_{t} +
p^{-1}_{t}y_{t}\right)-\left(x_{0} + p^{-1}_{0}y_{0}\right)}{\left(x_{0} +
p^{-1}_{0}y_{0}\right)} \\
& \text{IL borrowed}^{(x)}(p^{-1}_{t}) = \frac{\left(x_{t} + p^{-1}_{t}y_{t}
\right) - \left(x_{0} + p^{-1}_{t}y_{0}\right)}{\left(x_{0} +
p^{-1}_{0}y_{0}\right)}.
\end{split}%
\end{equation}

\subsection{Payoff of IL Protection Claim}

We fix the maturity time $T$.

\begin{definition}[Payoff of IL protection claim]
We define the protection claim against the IL as a derivative security whose
pay-off at time $T$ equals the negative value of the IL. For the funded LP,
the pay-off at time $T$ is defined by 
\begin{equation}
\mathrm{Payoff \ funded}(p_{T})=-\mathrm{IL \ funded}^{(y)}(p_{T}).  \label{eq:il_funded_p}
\end{equation}%
For the borrowed LP, the pay-off at time $T$ is defined by 
\begin{equation}
\mathrm{Payoff \ borrowed}(p_{T})=-\mathrm{IL \ borrowed}^{(y)}(p_{T}).  \label{eq:il_borrowed_p}
\end{equation}
\end{definition}

A liquidity provider of staked LP with notional $N^{(y)}$ can buy the IL
protection claim to eliminate impermanent loss of the staked LP. By Eq. \eqref{eq:V3aaa}, at time $T$ the P\&L of holder's LP will be matched by the
pay-off of the IL protection claim in Eq. \eqref{eq:il_funded_p} or Eq. %
\eqref{eq:il_borrowed_p}. As a result, the liquidity provider can perfectly
hedge the IL at time $T$.

\section{Applications to Uniswap AMM Protocol}

\label{sec:apps}

We now derive explicit formulas for the IL of funded and borrowed LP stakes
under the Uniswap V2 and V3 protocols.

\subsection{Uniswap V2}

\label{sec:V2}

In Uniswap V2 (\cite{V2}), the CFMM is defined by the constant product rule
as follows 
\begin{equation}  \label{eq:v1}
x' y' = L'^{2}.
\end{equation}
where $x'$ and $y'$ are pool reserves and $L'$ is the pool liquidity parameter.
In the Uniswap V2 white paper, the constant is defined by $k$. We use $L'^{2}$%
, $L'>0$, according to the specification of the V3 protocol.

The pool price is determined by pool reserves as follows 
\begin{equation}  \label{eq:v1a}
p \equiv \frac{y'}{x'}
\end{equation}
Thus, we need to solve Eq. (\ref{eq:v1}) and (\ref{eq:v1a}) in the two
unknowns $x'$ and $y'$. Substituting $y'=p x'$ from Eq. (\ref{eq:v1a}) into Eq. (%
\ref{eq:v1}), we obtain that the LP stakes are given as follows 
\begin{equation}  \label{eq:V2}
x' = \sqrt{\frac{L'^{2}}{p}} , \ y' = \sqrt{p L'^{2}}.
\end{equation}

We now consider the LP position with staked units $x$ and $y$ and the liquidity constant $L$. From Eq. \eqref{eq:V2}, the value of LP position is given by\footnote{%
This relationship follows from the condition that the internal price in Eq. %
\eqref{eq:v1a} inferred by pool reserves follows an external price $p_t$,
observed on other DEXes and centralised exchanges. In practice, the reserves of the
liquidity pools are balanced so that the internal price in Eq. \eqref{eq:v1a} follows the external price feeds within tight bands most of the time due to
arbitrage operations of multiple arbitrageurs in the blockchain ecosystem. For
details of such arbitrages see among others \citet{Milionis2022}, %
\citet{Park2023}, \citet{Lehar2024}, \citet{Cartea2024}, \citet{Cartea2024a}.
Same considerations apply for the internal price in Eq. \eqref{eq:V32a}
implied by pool reserves for Uniswap V3 pools.}:
\begin{equation}  \label{eq:V3a}
\begin{split}
V^{(y)}_{t} & = p_{t} x_{t} + y_{t} = 2 L \sqrt{p_{t}}.
\end{split}%
\end{equation}

\begin{proposition}[Funded LP]
The funded P\&L in Eq. (\ref{eq:pnl_funded}) is computed by: 
\begin{equation}  \label{eq:pnl_funded_V2}
\mathrm{P\&L \ funded}^{(y)}(p_{t}) = 2 L \sqrt{p_{0}} \left( \sqrt{\frac{p_{t}}{p_{0}%
}} -1 \right).
\end{equation}

The nominal IL for funded LP in Eq. (\ref{eq:il_funded}) is computed by 
\begin{equation}  \label{eq:il_funded_V2}
\mathrm{Nom \ IL \ funded}^{(y)}(p_{t}) = \sqrt{\frac{p_{t}}{p_{0}}} -1.
\end{equation}
\end{proposition}

\begin{proof}
Using Eq. \eqref{eq:V3a}, we obtain
\begin{equation} \label{eq:V3aa} 
\begin{split}
\mathrm{P\&L \ funded}^{(y)}(p_{t}) & = 2 L \sqrt{p_{t}} - 2 L \sqrt{p_{0}}  = 2 L \sqrt{p_{0}}  \left( \sqrt{\frac{p_{t}}{p_{0}}} -1 \right).
\end{split}
\end{equation}
Given an initial notional of the stake such as $N^y=V_{0}=2L\sqrt{p_{0}}$, we obtain the nominal IL.
\end{proof}

\begin{proposition}[IL for Borrowed LP in Uniswap V2]
The borrowed P\&L in Eq. (\ref{eq:pnl_borrowed}) is given by: 
\begin{equation}  \label{eq:pnl_borrowed_V2}
\mathrm{P\&L\ borrowed}^{(y)}(p_{t}) = - L \sqrt{p_{0}} \left( \sqrt{\frac{p_{t}}{%
p_{0}}} - 1\right)^{2}.
\end{equation}
The nominal IL for borrowed LP in Eq. (\ref{eq:il_borrowed_nom}) is given by 
\begin{equation}  \label{eq:il_borrowed_nom_V2}
\mathrm{Nom \ IL \ borrowed}^{(y)}(p_{t}) = - \frac{1}{2} \left( \sqrt{\frac{p_{t}}{%
p_{0}}} - 1\right)^{2}.
\end{equation}

Relative IL in Eq. (\ref{eq:il_borrowed_rel}) is given by 
\begin{equation}  \label{eq:il_borrowed_rel_V2}
\mathrm{Rel \ IL \ borrowed}^{(y)}(p_{t}) = - \frac{ \left( \sqrt{\frac{p_{t}}{p_{0}}}
- 1\right)^{2}}{ \frac{p_{t}}{p_{0}} +1}.
\end{equation}
\end{proposition}

\begin{proof}
From Eq. \eqref{eq:V2} we note that 
\begin{equation} \label{eq:v4a} 
p_{t} x_{0} + y_{0}= p_{t} \sqrt{\frac{L^{2}}{p_{0}}}+ \sqrt{p_{0} L^{2}} = \sqrt{p_{0}} L\left( \frac{p_{t}}{p_{0}} +1 \right).
\end{equation}
Thus, we obtain
\begin{equation} \label{eq:v4} 
\begin{split}
\mathrm{P\&L \ borrowed}^{(y)}(p_{t}) &  = \left(p_{t} x_{t} + y_{t}\right)  - \left(p_{t} x_{0} + y_{0}\right)\\
& = 2\sqrt{p_{t} L^{2}} - \sqrt{p_{0}} L\left( \frac{p_{t}}{p_{0}} +1 \right)\\
& = - L \sqrt{p_{0}}  \left( \sqrt{\frac{p_{t}}{p_{0}}} - 1\right)^{2}.
\end{split}
\end{equation}

\end{proof}

It is clear that the maximum is $0$ at $p_{t}=p_{0}$, and the
borrowed P\&L is negative for any other value of $p_{t}$.

\begin{corollary}[Payoff of claim for IL protection for Uniswap V2 AMM]
Using definitions in Eq. \eqref{eq:il_funded_p} and Eq. %
\eqref{eq:il_borrowed_p} for pay-offs of protection claim against funded and
borrowed LP, respectively, along with respective Eqs \eqref{eq:il_funded_V2}
and \eqref{eq:il_borrowed_nom_V2}, we obtain 
\begin{equation}  \label{eq:il_payoff_v2}
\begin{split}
& \mathrm{Payoff \ funded}(p_{T}) = 1 - \sqrt{\frac{p_{T}}{p_{0}}} , \\
& \mathrm{Payoff \ borrowed}(p_{T}) = \frac{1}{2} \left( \sqrt{\frac{p_{T}}{p_{0}}} -
1\right)^{2}.
\end{split}%
\end{equation}
\end{corollary}

In subplot (A) of Figure \ref{V2il}, we show ETH units (left y-axis) and
USDT units (right y-axis) for LP Uniswap V2 with $1m$ USDT notional and $%
p_{0}=2000$ ETH/USDT price. The initial LP units of (ETH, USDT) are $(250,
500000)$. The red bar at $p=1500$ shows LP units of (289, 433013) with the
higher allocation to ETH units as ETH/USDT price falls. The green bar at $%
p=2500$ shows corresponding LP units of $(224, 559017)$ with the higher
allocation to USDT units as ETH/USDT price rises. In subplot (B), we show
USDT values of 50\%/50\% ETH/USDT portfolio, Funded LP position and Borrowed
LP position. The funded LP underperforms the 50\%/50\% portfolio on both the
upside (because LP position reduces ETH units) and on the downside (because
LP position increases ETH units). The value of the borrowed LP has zero
first-order beta to ETH with negative quadratic convexity to ETH/USDT
changes.

\begin{figure}[]
\begin{center}
\includegraphics[width=1.0\textwidth, angle=0]
{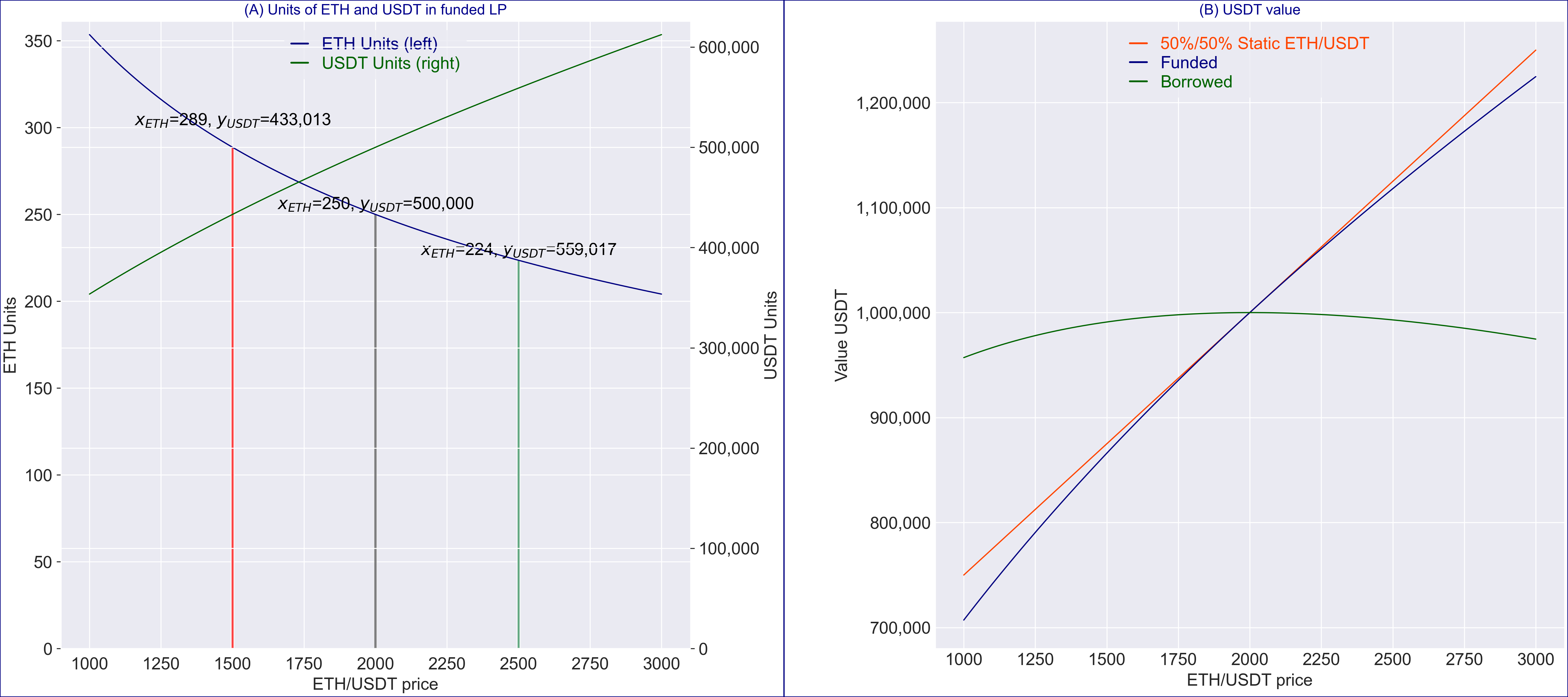}\vspace*{-1.\baselineskip}
\end{center}
\caption{(A) ETH units (left y-axis) and USDT units (right y-axis) for LP
Uniswap V2. (B) USDT value of 50\%/50\% ETH/USDT portfolio, Funded LP
position and Borrowed LP position. Uniswap V2 LP position is constructed
using $1m$ USDT notional with $p_{0}=2000$ ETH/USDT price.}
\label{V2il}
\end{figure}

\subsection{Application to Uniswap V3}

\label{sec:V3}

In Uniswap V3 protocol (see \cite{V3}), the staked liquidity $(x, y)$ of LP satisfies the following equation:
\begin{equation}  \label{eq:V31}
x_{v} y_{v} = L^{2},
\end{equation}
where $x_{v}$ and $y_{v}$ are termed as virtual reserves: 
\begin{equation}  \label{eq:V32}
x_{v} \equiv x + \frac{L}{\sqrt{p_{b}}}, \ y_{v} \equiv y + L\sqrt{p_{a}},
\end{equation}
with the liquidity amount $L$ provided in the price range $[p_{a}, p_{b}]$.

For ETH/USDT pool, the price $p$ is set by: 
\begin{equation}  \label{eq:V32a}
p \equiv \frac{y_{v}}{x_{v}} = \frac{y + L\sqrt{p_{a}}}{x+\frac{L}{\sqrt{%
p_{b}}}}.
\end{equation}

Eqs (\ref{eq:V31}) and (\ref{eq:V32a}) are viewed as two equations in four
unknowns ($x$, $y$, $L$, $p$): 
\begin{equation}  \label{eq:V33}
\begin{split}
& \left( x + \frac{L}{\sqrt{p_{b}}}\right)\left( y + L\sqrt{p_{a}} \right) =
L^{2} \\
& p = \frac{y + L\sqrt{p_{a}}}{x+\frac{L}{\sqrt{p_{b}}}}.
\end{split}%
\end{equation}
The above equations have to hold at any point in time, with internal price $%
p $ tracking an external price observed on other trading venues within tight
bounds because arbitrageurs adjust pool reserves accordingly to eliminate
arbitrage opportunities. Adding (removing) liquidity amounts to increasing
(decreasing) $L$ through increasing $x$ and $y$, while keeping $p$ constant.
Swapping (trading) tokens amounts to changing $x$, $y$ and $p$,
while keeping $L$ the same.

Next,  we gather several technical results related to this ``mechanics".
Those follow from the whitepaper \cite{V3} in a relatively straightforward manner. Similar derivations can be found in \cite{Elsts2021}.
Firstly, we solve for $x$ and $y$ given $L$ and $p$ as independent
variables. For a given position with $L$ and $p$ as external parameters,
this solution provides how many units $x$ and $y$ are assigned to the LP
position.

\begin{lemma}[Solution for $x$ and $y$]
Thus for $p \in (p_{a}, p_{b})$, the LP units $x$ and $y$ are given by: 
\begin{equation}  \label{eq:V34}
x = L\left(\frac{1}{\sqrt{p}} - \frac{1}{\sqrt{p_{b}}}\right), \ y = L \left(%
\sqrt{p} - \sqrt{p_{a}}\right).
\end{equation}

For $p \leq p_{a}$, the position is fully in token 1: 
\begin{equation}  \label{eq:V34a}
x = L\left(\frac{1}{\sqrt{p_{a}}} - \frac{1}{\sqrt{p_{b}}}\right), \ y = 0.
\end{equation}

For $p \geq p_{b}$, the position is fully in token 2: 
\begin{equation}  \label{eq:V34b}
x = 0, \ y = L \left(\sqrt{p_{b}} - \sqrt{p_{a}}\right).
\end{equation}
\end{lemma}

\begin{proof}
We substitute the second equation in Eq. (\ref{eq:V33})
\begin{equation} \label{eq:V33a} 
 \left(y + L\sqrt{p_{a}}\right) = p \left( x+\frac{L}{\sqrt{p_{b}}}\right)
\end{equation}
into the first one to obtain
 \begin{equation} \label{eq:V33aa} 
 \left( x + \frac{L}{\sqrt{p_{b}}}\right)^{2} = \frac{L^{2}}{p} \ \Rightarrow  x = L\left(\frac{1}{\sqrt{p}} - \frac{1}{\sqrt{p_{b}}}\right)
\end{equation}
and
\begin{equation} \label{eq:V33b} 
 \left(y + L\sqrt{p_{a}}\right) = L \sqrt{p} \ \Rightarrow y = L \left(\sqrt{p} - \sqrt{p_{a}}\right).
\end{equation}

\end{proof}

We obtain that the initial value of LP using Eq. (\ref{eq:V34}) for $p_{0}
\in (p_{a}, p_{b})$ is given by 
\begin{equation}  \label{eq:V3p}
V_{0} \equiv p_{0} x_{0} + y_{0}= L \bar{V}_{0}, \quad \bar{V}_{0} \equiv 2\sqrt{p_{0}} - \frac{p_{0}}{%
\sqrt{p_{b}}} - \sqrt{p_{a}},
\end{equation}
where $\bar{V}_{0}$ is the initial value of the LP per unit of liquidity.

\begin{corollary}[Initial notional $N^y$]
Given an initial notional of the stake such as $N^y=V_{0}$, using Eq. (\ref%
{eq:V3p}) we obtain that the provided liquidity $L$ is set by 
\begin{equation}  \label{eq:V3n2}
\begin{split}
& L = \frac{N^y}{2\sqrt{p_{0}} - \frac{p_{0}}{\sqrt{p_{b}}} - 
\sqrt{p_{a}} }.
\end{split}%
\end{equation}
\end{corollary}

\subsubsection{Impermanent Loss}

\begin{proposition}[Funded P\&L and IL in Uniswap V3]
\label{prop:pnl_funded_V3} The P\&L of the funded position in Eq. (\ref%
{eq:pnl_funded}) at time $t$ with current price $p_{t}$ is given by 
\begin{equation}  \label{eq:pnl_funded_V3}
\mathrm{P\&L \ funded}^{(y)} = 
\begin{cases}
L \left[ 2\left(\sqrt{p_{t}} -\sqrt{p_{0}} \right) + \frac{p_{0}-p_{t}}{%
\sqrt{p_{b}}}\right] \quad & \, p_{t} \in (p_{a}, p_{b}) \\ 
L \left[ p_{t} \left(\frac{1}{\sqrt{p_{a}}} - \frac{1}{\sqrt{p_{b}}}\right)
+ \frac{p_{0}}{\sqrt{p_{b}}} - 2\sqrt{p_{0}} + \sqrt{p_{a}}\right] \quad & 
p_{t} \leq p_{a} \\ 
L\left[ \sqrt{p_{b}} + \frac{p_{0}}{\sqrt{p_{b}}} - 2\sqrt{p_{0}}\right]
\quad & p_{t} \geq p_{b},
\end{cases}%
\end{equation}
where 
\begin{equation}  \label{eq:pnl_funded_V3n}
L = \frac{N^{y}}{2\sqrt{p_{0}} - \frac{p_{0}}{\sqrt{p_{b}}} - \sqrt{p_{a}} }
\end{equation}
and $N^{y}$ is USDT notional.

The nominal IL for funded position defined in Eq. (\ref{eq:il_funded}) is
computed by 
\begin{equation}  \label{eq:il_funded_V3}
\begin{split}
&\mathrm{Nom \ IL \ funded}^{(y)} = \frac{\mathrm{P\&L \ funded}^{(y)}}{L \left( 2\sqrt{p_{0}}
- \frac{p_{0}}{\sqrt{p_{b}}} - \sqrt{p_{a}} \right)}.
\end{split}%
\end{equation}
\end{proposition}

\begin{corollary}[The pay-off of IL protection claim for funded LP]
The pay-off of IL protection claim at maturity date $T$ for funded LP in Eq. %
\eqref{eq:il_funded_p} is given by the following compact formula 
\begin{equation}  \label{eq:il_funded_V3a}
\mathrm{Payoff \ funded} (p_{t}) = - \frac{ \frac{p_{t}}{\sqrt{f(p_{t}; p_{a}, p_{b})}%
} +\sqrt{f(p_{t}; p_{a}, p_{b})} - \frac{p_{t}}{\sqrt{p_{b}}} -\sqrt{p_{a}}%
} { \frac{p_{0}}{\sqrt{f(p_{0}; p_{a}, p_{b})}} +\sqrt{f(p_{0}; p_{a}, p_{b})%
} - \frac{p_{0}}{\sqrt{p_{b}}} -\sqrt{p_{a}} } + 1,
\end{equation}
where $f(x; p_{a}, p_{b})=\max\left(\min\left(x, p_{b}\right), p_{a}\right)$.
\end{corollary}

\begin{proof}
See Appendix \ref{ap_prop:pnl_funded_V3}.
\end{proof}

\begin{proposition}[P\&L of borrowed LP]
\label{prop:pnl_borrowed_V3}

The P\&L of the borrowed position in Eq. (\ref{eq:pnl_borrowed}) is given by 
\begin{equation}  \label{eq:pnl_borrowed_V3}
\mathrm{P\&L\ borrowed}^{(y)}(p_{t}) = 
\begin{cases}
- L \sqrt{p_{0}} \left(\sqrt{\frac{p_{t}}{p_{0}}}-1\right)^{2} \quad & \,
p_{t} \in (p_{a}, p_{b}) \\ 
L \left[ p_{t} \left(\frac{1}{\sqrt{p_{a}}} - \frac{1}{\sqrt{p_{0}}}\right)
- \left(\sqrt{p_{0}} - \sqrt{p_{a}}\right) \right] \quad & p_{t} \leq p_{a}
\\ 
L\left[ \left(\sqrt{p_{b}} - \sqrt{p_{0}}\right)- p_{t} \left(\frac{1}{\sqrt{%
p_{0}}} - \frac{1}{\sqrt{p_{b}}}\right) \right] \quad & p_{t} \geq p_{b}.
\end{cases}%
\end{equation}

Nominal borrowed impermanent loss in Eq. (\ref{eq:il_borrowed_nom}) is given
by 
\begin{equation}  \label{eq:il_borrowed_nom_V3}
\begin{split}
&\mathrm{Nom \ IL \ borrowed}^{(y)}(p_{t}) = \frac{\mathrm{P\&L \ borrowed}^{(y)}(p_{t})}{L
\left( 2\sqrt{p_{0}} - \frac{p_{0}}{\sqrt{p_{b}}} - \sqrt{p_{a}} \right)}.
\end{split}%
\end{equation}
\end{proposition}

\begin{proof}
See Appendix \ref{ap_prop:pnl_borrowed_V3}.
\end{proof}

\begin{corollary}[Payoff of IL protection claim for borrowed LP]
The pay-off of IL protection claim at maturity date $T$ for borrowed LP in Eq. %
\eqref{eq:il_borrowed_p} is given by the following compact formula 
\begin{equation}  \label{eq:il_borrowed_nom_V3a}
\mathrm{ Payoff \ borrowed}(p_{t}) = - \frac{ \frac{p_{t}}{\sqrt{f(p_{t}; p_{a}, p_{b})}}
+\sqrt{f(p_{t}; p_{a}, p_{b})} - \frac{p_{t}}{\sqrt{f(p_{0}; p_{a}, p_{b})}}
-\sqrt{f(p_{0}; p_{a}, p_{b})}} { \frac{p_{0}}{\sqrt{f(p_{0}; p_{a}, p_{b})}}
+\sqrt{f(p_{0}; p_{a}, p_{b})} - \frac{p_{0}}{\sqrt{p_{b}}} -\sqrt{p_{a}} },
\end{equation}
where $f(x; p_{a}, p_{b})=\max\left(\min\left(x, p_{b}\right), p_{a}\right)$.
\end{corollary}

In subplot (A) of Figure \ref{V3il}, we show ETH units (left y-axis) and
USDT units (right y-axis) for LP on Uniswap V3 with $1m$ USDT notional and $%
p_{0}=2000$, $p_{a}=1500$, $p_{b}=2500$. The initial LP units of (ETH, USDT)
are $(220, 559282)$. The red bar at $p=1500$ shows LP units of (543, 0) with the
LP fully in ETH units when the price falls below the lower threshold $p_{a}$. The
green bar at $p=2500$ shows the corresponding LP units of $(0, 1052020)$ with the LP
fully in USDT units when the price rises above the upper threshold $p_{b}$. In
subplot (B), we show USDT values of 50\%/50\% ETH/USDT portfolio, Funded LP
positions and Borrowed LP positions. The value profile of the funded LP
resembles the profile of a covered call option (long ETH and short
out-of-the-money call). The value of the borrowed LP resembles the pay-off of
a short straddle (short both at-the-money call and put).

\begin{figure}[]
\begin{center}
\includegraphics[width=1.0\textwidth, angle=0]
{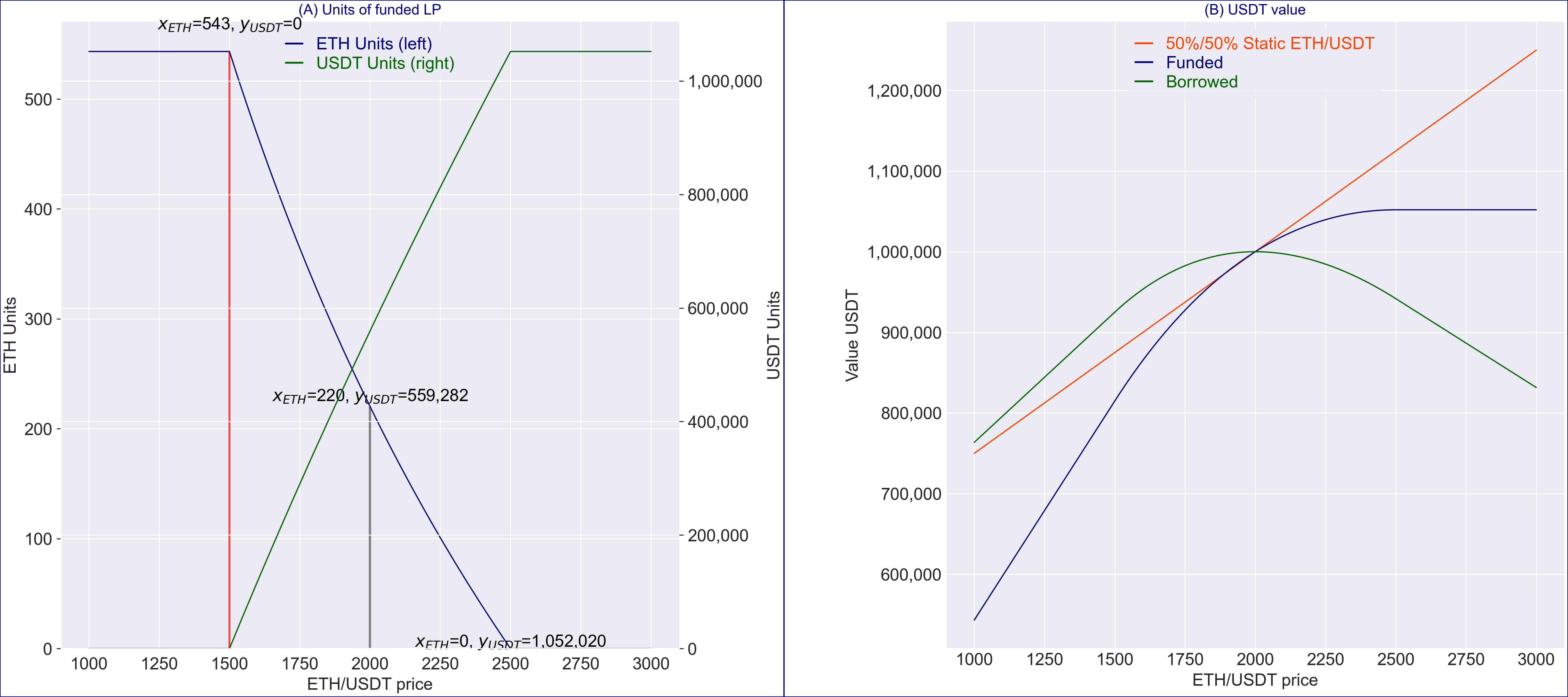}\vspace*{-1.\baselineskip}
\end{center}
\caption{(A) ETH units (left y-axis) and USDT units (right y-axis) for LP on
Uniswap V3. (B) USDT value of 50\%/50\% ETH/USDT portfolio, Funded LP
position and Borrowed LP position. Uniswap V3 LP position is constructed
using $1m$ USDT notional with $p_{0}=2000$, $p_{a}=1500$, $p_{b}=2500$.}
\label{V3il}
\end{figure}

In Figure \ref{lp_pnl_ranges}, we show P\&L profiles of borrowed and funded
LPs as functions of ranges for Uniswap V3 and the full range for Uniswap V2. For
borrowed LPs, narrow ranges result in higher losses at the same price levels.
For funded LPs, narrower ranges result in higher downside losses and lower
upside potential. 
\begin{figure}[]
\begin{center}
\includegraphics[width=0.8\textwidth, angle=0]
{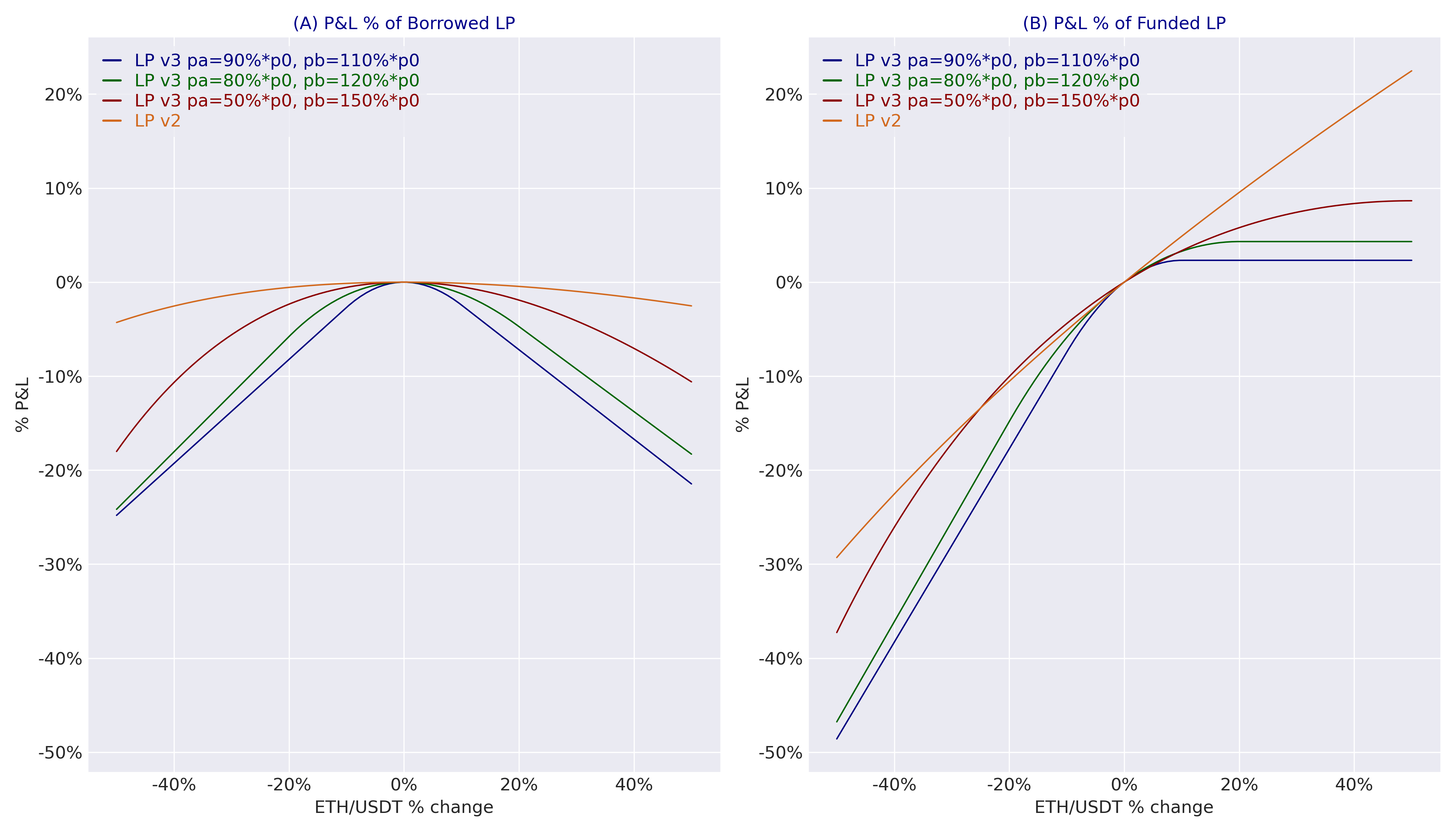}\vspace*{-1.\baselineskip}
\end{center}
\caption{P\&L $\%$ of hedged and funded LPs as functions of ranges for
Uniswap V3 and full range for Uniswap V2. (A) P\&L $\%$ of borrowed LP; (B)
P\&L $\%$ of funded LP. LP positions are constructed using $1m$ USDT
notional with $p_{0}=2000$.}
\label{lp_pnl_ranges}
\end{figure}

\subsection{Decomposition of the IL under Uniswap V3 into Simple Payoffs}

\label{sec:ild}

First we derive the model-independent decomposition of IL into pay-offs of
vanilla and ``exotic'' options. We will further apply these decompositions
for model-based valuation of IL protection claims in Uniswap V3.

\begin{proposition}[Decomposition of IL for Funded LP]
\label{prop:il_dec_funded1} IL of funded LP can be decomposed into the three
parts as follows: 
\begin{equation}  \label{eq:il_dec_funded1}
\mathrm{IL \ funded}^{(y)}(p_{t}) = \frac{1}{\bar{V}_{0}}\left[u_{0}(p_{t}) + 2u_{1/2}(p_{t}) + u_{1}(p_{t})\right],
\end{equation}
where $\bar{V}_{0}$ is given in Eq. \eqref{eq:V3p}, and $u_{0}(p_{t})$, $u_{1/2}(p_{t})$, and $u_{1}(p_{t})$ are linear part,
(exotic) square root price, and (vanilla) option part defined as follows 
\begin{equation}  \label{eq:il_dec_funded2}
\begin{split}
& u_{0}(p_{t}) = - \frac{1}{\sqrt{p_{b}}} p_{t} + \left(\frac{p_{0}}{\sqrt{%
p_{b}}}-2\sqrt{p_{0}}\right) \\
& u_{1/2}(p_{t}) = \sqrt{p_{t}} \mathbbm{1}\left\{ p_{a} < p_{t} <
p_{b}\right\} \\
& u_{1}(p_{t}) = -\frac{1}{\sqrt{p_{a}}} \max\left\{ p_{a}-p_{t}, 0\right\}
+ \frac{1}{\sqrt{p_{b}}} \max\left\{ p_{t}-p_{b}, 0\right\} + 2\sqrt{p_{a}} %
\mathbbm{1}\left\{p_{t}\leq p_{a}\right\} + 2\sqrt{p_{b}}\mathbbm{1}%
\left\{p_{t} \geq p_{b}\right\}.
\end{split}%
\end{equation}
\end{proposition}

\begin{proof}
See Appendix \ref{ap_prop:il_dec_funded1}.
\end{proof}

\begin{proposition}[Decomposition of IL for Borrowed LP]
\label{prop:il_dec_borrowed1} IL of borrowed LP can be decomposed into the
three parts as follows: 
\begin{equation}  \label{eq:il_dec_borrowed1}
\mathrm{IL \ borrowed}^{(y)}(p_{t}) = \frac{1}{\bar{V}_{0}}\left[u_{0}(p_{t}) + 2u_{1/2}(p_{t}) + u_{1}(p_{t})\right],
\end{equation}
where $\bar{V}_{0}$ is given in Eq. \eqref{eq:V3p}, and $u_{0}(p_{t})$, $u_{1/2}(p_{t})$, and $u_{1}(p_{t})$ are linear part,
(exotic) square root price, and (vanilla) option part defined as follows 
\begin{equation}  \label{eq:il_dec_borrowed2}
\begin{split}
& u_{0}(p_{t}) = - \frac{1}{\sqrt{p_{0}}} p_{t} - \sqrt{p_{0}} \\
& u_{1/2}(p_{t}) = \sqrt{p_{t}} \mathbbm{1}\left\{ p_{a} < p_{t} <
p_{b}\right\} \\
& u_{1}(p_{t}) = -\frac{1}{\sqrt{p_{a}}} \max\left\{ p_{a}-p_{t}, 0\right\}
+ \frac{1}{\sqrt{p_{b}}} \max\left\{ p_{t}-p_{b}, 0\right\} + 2\sqrt{p_{a}} %
\mathbbm{1}\left\{p_{t}\leq p_{a}\right\} + 2\sqrt{p_{b}}\mathbbm{1}%
\left\{p_{t} \geq p_{b}\right\}.
\end{split}%
\end{equation}
\end{proposition}

\begin{proof}
See Appendix \ref{ap_prop:il_dec_borrowed1}.
\end{proof}

To summarise, we note that the only difference between the decomposition of
IL of the funded LP in Eq. \eqref{eq:il_dec_funded1} and the borrowed LP in Eq. %
\eqref{eq:il_dec_borrowed1} is given by the linear term $u_{0}(p_{t})$ with
the square root and option terms being the same. As a result, we can unify
our results for ILs of funded LP and borrowed LP as defined in Eqs.\eqref{eq:il_funded} and \eqref{eq:il_borrowed_nom}, respectively, for
Uniswap V3 for the following.

\begin{corollary}[Decomposition of IL for Funded and Borrowed LPs]
Using Eq. \eqref{eq:il_dec_funded1} for IL of funded LP and Eq. %
\eqref{eq:il_dec_borrowed1} for IL of borrowed LP we obtain 
\begin{equation}  \label{eq:il_payoff}
\begin{split}
& \mathrm{IL \ funded}^{(y)}(p_{t}) = \frac{1}{\bar{V}_{0}}\left[u^{funded}_{0}(p_{t}) + 2u_{1/2}(p_{t}) +
u_{1}(p_{t})\right], \\
& \mathrm{IL \ borrowed}^{(y)}(p_{t}) = \frac{1}{\bar{V}_{0}}\left[u^{borrowed}_{0}(p_{t}) + 2u_{1/2}(p_{t}) +
u_{1}(p_{t})\right],
\end{split}%
\end{equation}
where 
\begin{equation}  \label{eq:il_payoff2}
\begin{split}
& u^{funded}_{0}(p_{t}) = - \frac{1}{\sqrt{p_{b}}} p_{t} + \left(\frac{p_{0}%
}{\sqrt{p_{b}}}-2\sqrt{p_{0}}\right) \\
& u^{borrowed}_{0}(p_{t}) = - \frac{1}{\sqrt{p_{0}}} p_{t} - \sqrt{p_{0}} \\
& u_{1/2}(p_{t}) = \sqrt{p_{t}} \mathbbm{1}\left\{ p_{a} < p_{t} <
p_{b}\right\} \\
& u_{1}(p_{t}) = -\frac{1}{\sqrt{p_{a}}} \max\left\{ p_{a}-p_{t}, 0\right\}
+ \frac{1}{\sqrt{p_{b}}} \max\left\{ p_{t}-p_{b}, 0\right\} + 2\sqrt{p_{a}} %
\mathbbm{1}\left\{p_{t}\leq p_{a}\right\} + 2\sqrt{p_{b}}\mathbbm{1}%
\left\{p_{t} \geq p_{b}\right\}.
\end{split}%
\end{equation}
\end{corollary}

\begin{corollary}[Payoff of IL protection claim for Uniswap V3 AMM]
Using the definitions in Eq. \eqref{eq:il_funded_p} and Eq. %
\eqref{eq:il_borrowed_p} for pay-offs of IL protection claim of the funded and
borrowed LP, respectively, at maturity time $T$ along with the decomposition of
IL in Eq. \eqref{eq:il_payoff}, we obtain 
\begin{equation}  \label{eq:il_payoff3}
\begin{split}
& \mathrm{Payoff \ funded}(p_{T}) = - \frac{1}{\bar{V}_{0}}\left[u^{funded}_{0}(p_{T}) + 2u_{1/2}(p_{T}) +
u_{1}(p_{T})\right], \\
& \mathrm{Payoff \ borrowed}(p_{T}) = - \frac{1}{\bar{V}_{0}}\left[u^{borrowed}_{0}(p_{T}) +
2u_{1/2}(p_{T}) + u_{1}(p_{T})\right],
\end{split}%
\end{equation}
where functions $u$ are defined in Eq. \eqref{eq:il_payoff2}.
\end{corollary}

\section{Static Replication of Impermanent Loss with Vanilla Options}

\label{sec:sh}

We derive a static replication of the IL at fixed maturity time $T$
using a portfolio of European call and put options. Consequently, the IL can
be hedged by buying a portfolio of traded options.

We introduce pay-off functions of call and put options as follows\footnote{%
We note that on crypto exchanges most options are the so-called inverse
options with the pay-off paid in the underlying token. There is a direct
arbitrage-based equivalence between vanilla and inverse options, see %
\citet{Alexander2023} and \citet{LucicSepp2024} for details, so that our
analysis follows the same logic when using inverse options.}: 
\begin{equation}  \label{eq:sr0}
u^{call}(p, k)=(p-k)^{+}, \ u^{put}(p, k)=(k-p)^{+},
\end{equation}
where $k$ is the strike price $k$ and $p$ is the current price.

We note that \cite{Deng2023} and \cite{MaireWunsch2024} obtain the following
replication formula for the replication of funded P\&L as defined in Eq. %
\eqref{eq:pnl_funded} and its analytic expression for Uniswap V3 defined in
Eq. \eqref{eq:pnl_funded_V3} (expressed using our notation) 
\begin{equation}  \label{eq:sr1}
\begin{split}
\frac{1}{L} \ \mathrm{P\&L \ funded}^{(y)}(p_{t}) + \bar{V}_{0} & = -\frac{1}{4}
\int^{p_{b}}_{p_{a}}k^{-3/2}\left(u^{put}(p_{t}, k)+u^{call}(p_{t},
k)\right)dk \\
& + \frac{1}{2\sqrt{p_{a}}}\left(u^{call}(p_{t}, p_{a})-u^{put}(p_{t},
p_{a})\right) \\
& + \frac{1}{2\sqrt{p_{b}}}\left(u^{put}(p_{t},
p_{b})-u^{call}(p_{t}, p_{b})\right).
\end{split}%
\end{equation}

\cite{Deng2023} derive the replication formula \eqref{eq:sr1} using the
following representation of IL under Uniswap V3 
\begin{equation}  \label{eq:sr2}
\begin{split}
\frac{1}{L} \ \mathrm{P\&L \ funded}^{(y)}(p_{t}) + \bar{V}_{0} & = p_{t}\left( \frac{1}{\sqrt{%
p_{t}}} - \frac{1}{\sqrt{p_{b}}} \right)^{+} - p_{t}\left( \frac{1}{\sqrt{%
p_{t}}} - \frac{1}{\sqrt{p_{a}}} \right)^{+} \\
& + \left( \sqrt{p_{t}} - \sqrt{p_{a}}\right)^{+} -\left( \sqrt{p_{t}} - 
\sqrt{p_{b}}\right)^{+}.
\end{split}%
\end{equation}
where we apply our definition of nominal IL in Eq. \eqref{eq:il_funded_V3}
and Carr-Madan representation (\citet{Carr2001}). We note that since the
function on the left-hand side in \eqref{eq:sr2} is not twice
differentiable, strictly speaking, the Carr-Madan representation does not
apply. For completeness, in Appendix \ref{sec:carrmadan} we provide a
derivation that only relies on the existence of the generalised derivatives.

Decomposition in Eq. \eqref{eq:sr2} includes four exotic pay-offs on the
square root of the price. In contrast, we derive an alternative
decomposition of IL of the funded LP in Eq. \eqref{eq:il_dec_funded1} which
decomposes the IL into one exotic pay-off on the square root of the price,
two pay-offs of vanilla call and put options, and two pay-offs of digital
options. It is clear that the decomposition of the IL is not unique and can
be done with a different base pay-off function. Our decomposition for IL of
funded and borrowed LPs in Eq. \eqref{eq:il_dec_funded1} and %
\eqref{eq:il_dec_borrowed1}, respectively, is most suitable for
model-dependent valuation of IL protection claims.

We note that there are two complications with the replication formula %
\eqref{eq:sr1} for practical usage. First, formula \eqref{eq:sr1} assumes
that there are strikes corresponding to lower and upper levels $p_{a}$ and $%
p_{b}$. In practice, on the Deribit exchange, the BTC and ETH options are traded
with the strike width of $1000$ and $50$ USD, respectively. Therefore, the
practical application of formula \eqref{eq:sr1} is very limited because it
is only applicable to LPs with ranges contained in strikes of traded options.
Second, the replication formula \eqref{eq:sr1} requires one to buy both calls and
puts at the same strikes. Thus, for strikes smaller than $p_{t}$, both
out-of-the-money puts and in-the-money calls are purchased and, for strikes
above $p_{t}$, both in-the-money puts and out-of-the-money calls are purchased.
In practice on the Deribit exchange, the liquidity for in-the-money calls and
puts is limited with much wider bid-ask spreads, so that implementation of
formula \eqref{eq:sr1} can be too cost-inefficient in reality.

Instead, we derive an alternative replication formula which only purchases
out-of-the-money puts for strikes below the current price and out-of-the-money
calls for strikes above the current price. We also note that our
replication formula applies to a generic CFMM of AMM protocols in addition to CFMMs
of Uniswap V2 or V3 protocols.

\subsection{Replication of IL with Vanilla Options}

We derive a replication portfolio for IL under a generic AMM protocol. We
assume that IL $IL(p)$ is a function of the current price $p$ and is specified
by Eq. \eqref{eq:il_funded} for funded LP or by Eq. \eqref{eq:il_borrowed_nom}
for borrowed LP. In particular, the IL for funded and borrowed LPs in
Uniswap V2 is obtained using Eq. (\ref{eq:il_funded_V2}) and (\ref{eq:il_borrowed_nom_V2}), respectively. For Uniswap V3, the IL for funded
and borrowed LPs is obtained using Eq. (\ref{eq:il_funded_V3}) and (\ref{eq:il_borrowed_nom_V3}), respectively. The corresponding IL for all these
specifications is denoted by $IL(p)$.

\begin{proposition}[Replication portfolio for generic LP]
\label{prop:replication_lp}

We consider nominal IL of a generic AMM as a function of price $IL(p)$, with $IL(p_{0})=0$. We fix maturity $T$ and construct a piecewise-linear approximation to the protection pay-off $-IL(p)$ using calls and puts of this maturity. For exact interpolation at the strikes, both grids include the initial LP price $p_{0}$ as their common anchor. This price is fixed by the LP entry and need not equal the $T$-forward price; the description as out-of-the-money puts and calls assumes zero carry, so that the inception forward equals $p_{0}$.

\textbf{Put side}. We assume a discrete grid of strikes $\mathcal{K}^{put}$
and corresponding pay-offs of put options $\mathcal{U}^{put}$: 
\begin{equation}  \label{eq:ph1}
\mathcal{K}^{put} = \left(k_{1}, k_{2},...,k_{N} \right), \ \mathcal{U}%
^{put} = \left(u_{1}, u_{2},...,u_{N} \right),
\end{equation}
where $u_n=(k_n-p)^+$, $n=1,..,N$, $k_{n-1} < k_{n}$ and $k_{N}=p_{0}$.
We consider the replication portfolio
of puts with weights $w_n$: 
\begin{equation}  \label{eq:ph2}
\begin{split}
\Pi^{put} = \sum^{N}_{n=1}w_{n}u^{put}_{n}.
\end{split}%
\end{equation}
We define the first-order derivative of the IL function at discrete strike
points as follows: 
\begin{equation}  \label{eq:ph4}
\delta IL(k_{n}) = \frac{IL(k_{n})-IL(k_{n-1})}{k_{n}-k_{n-1}}.
\end{equation}

Then the weights of put options for the replication portfolio in Eq. %
\eqref{eq:ph2} are computed by: 
\begin{equation}  \label{eq:ph5}
w_{n-1} = -\left(\delta IL(k_{n})- \delta IL(k_{n-1})\right), n=N, ...,3
\end{equation}
with $w_{N} = \delta IL(k_{N})$ and with $w_{1}=0$.

\textbf{Call side}. We assume a discrete grid of strikes $\mathcal{K}^{call}$ and
corresponding pay-offs of call options $\mathcal{U}^{call}$: 
\begin{equation}  \label{eq:phc1}
\mathcal{K}^{call} = \left(k_{1}, k_{2},...,k_{M} \right), \ \mathcal{U}%
^{call} = \left(u_{1}, u_{2},...,u_{M} \right),
\end{equation}
where $u_m=(p-k_m)^+$, $m=1,..,M$, with $k_{m} > k_{m-1}$ and $k_{1}=p_{0}$.
We consider the hedging portfolio of calls
with weights $w_m$: 
\begin{equation}  \label{eq:phc2}
\Pi^{call} = \sum^{M}_{m=1}w_{m}u^{call}_{m}.
\end{equation}
We compute the first-order derivative of the IL function $\delta IL(k_{m})$
as in Eq. \eqref{eq:ph4}.

Then the weights of call options for the replication portfolio in Eq. \eqref{eq:phc2} are computed by: 
\begin{equation}  \label{eq:phc3}
w_{m} = -\left(\delta IL(k_{m+1})- \delta IL(k_{m})\right), m=2, ...,M-1
\end{equation}
with $w_{1} = -\delta IL(k_{2})$ and with $w_{M}=0$. The resulting portfolio equals $-IL$ at every grid point and linearly interpolates between them. Outside the outermost strikes it extends the outer secant slopes. If the common anchor $p_{0}$ is unavailable, these weights alone do not reproduce the missing anchor value; an additional cash/forward term or an explicit approximation is required.
\end{proposition}

\begin{proof}
See Appendix \ref{sc:replication_lp}.
\end{proof}

\begin{corollary}[Static portfolio for replication of the pay-off of IL
Protection claim]
The option replication portfolio for pay-off of IL protection claim is given
by: 
\begin{equation}  \label{eq:il_static_rep}
\begin{split}
& \Pi \equiv \Pi^{put} +\Pi^{call} = \sum^{N}_{n=1}w_{n}u^{put}_{n} +
\sum^{M}_{m=1}w_{m}u^{call}_{m},
\end{split}%
\end{equation}
where the weights of put and call options are computed using Eq. %
\eqref{eq:ph5} and Eq. \eqref{eq:phc3}, respectively, with the IL function $IL(p)$
specified by Eqs. \eqref{eq:il_funded_V2} and \eqref{eq:il_borrowed_nom_V2} for Uniswap V2, or Eqs. \eqref{eq:il_funded_V3} and \eqref{eq:il_borrowed_nom_V3} for Uniswap V3. These are the nominal IL functions; the minus sign needed for protection is already included in the weights.

\end{corollary}

The cost of the replication portfolio is $\Pi_{0}$ computed using option
prices observed at the inception of an IL protection claim. We use ask prices of options to calculate the initial cost $\Pi_{0}$. Option markets on Deribit are relatively liquid when traded notionals are below $1000$ ETH. For less liquid options markets, we recommend incorporating estimated slippage costs for computing IL hedging portfolio.

In Figure \ref{il_replication_borrowed}, we illustrate the application of
formulas \eqref{eq:ph2} and \eqref{eq:phc2} for replicating of IL for
borrowed Uniswap V3 LP using $1m$ USDT notional, $p_{0}=2000$ ETH/USDT with $%
p_{a}=1500$ and $p_{b}=2500$. We use strikes with widths of $50$ USDT in
alignment with ETH options traded on Deribit exchange (for options with
maturity of less than 3 days, Deribit introduces new strikes with widths of $%
25$). In subplot (A), we show the IL of the borrowed LP position and the
pay-offs of replicating calls and puts portfolios (with negative signs to
align with the P\&L). In subplot (B), we show the residual computed as the
difference between the IL and the pay-off of the replication portfolios. In
Subplot (C), we show the number of put and call option contracts for the
replication portfolios.

From Eq. \eqref{eq:ph11} it is clear that the approximation error is zero at
strikes in the grid, which is illustrated in subplot (B). The maximum value
of the residual is $0.025\%$ or $2.5$ basis points, which is very small. A
small approximation error with a similar magnitude will occur in case $%
p_{0} $, $p_{a}$, $p_{b}$ are not placed exactly at the strike grid.

In Figure \ref{strike_hedging_error}, we show the residual hedging error for strike widths of $100$, $50$, and $25$ USDT. We note that for ETH options on the Deribit exchange, the strike width of listed options with maturities up to one month is $100$ USDT. For options with maturities up to one week and up to three days, Deribit adds additional strikes around at-the-money region with widths of $50$ and $25$ USDT, respectively. We notice that the magnitude of the residual error decreases when using narrow widths of strikes with $25$. However, even for a wider width of $100$, the residual error is less than $0.1\%$ of the notional. For options on Bitcoin traded on Deribit, the strike width varies between $500$ for options close to the spot and $5000$ for far out-of-the-money strikes (where the absolute value of option delta is less than $0.1$). We recommend setting the LP position in ranges where the deltas of the put and call options are less than $-0.2$ and greater than $0.2$, respectively.

\begin{figure}[H]
\begin{center}
\includegraphics[width=0.8\textwidth, angle=0]
{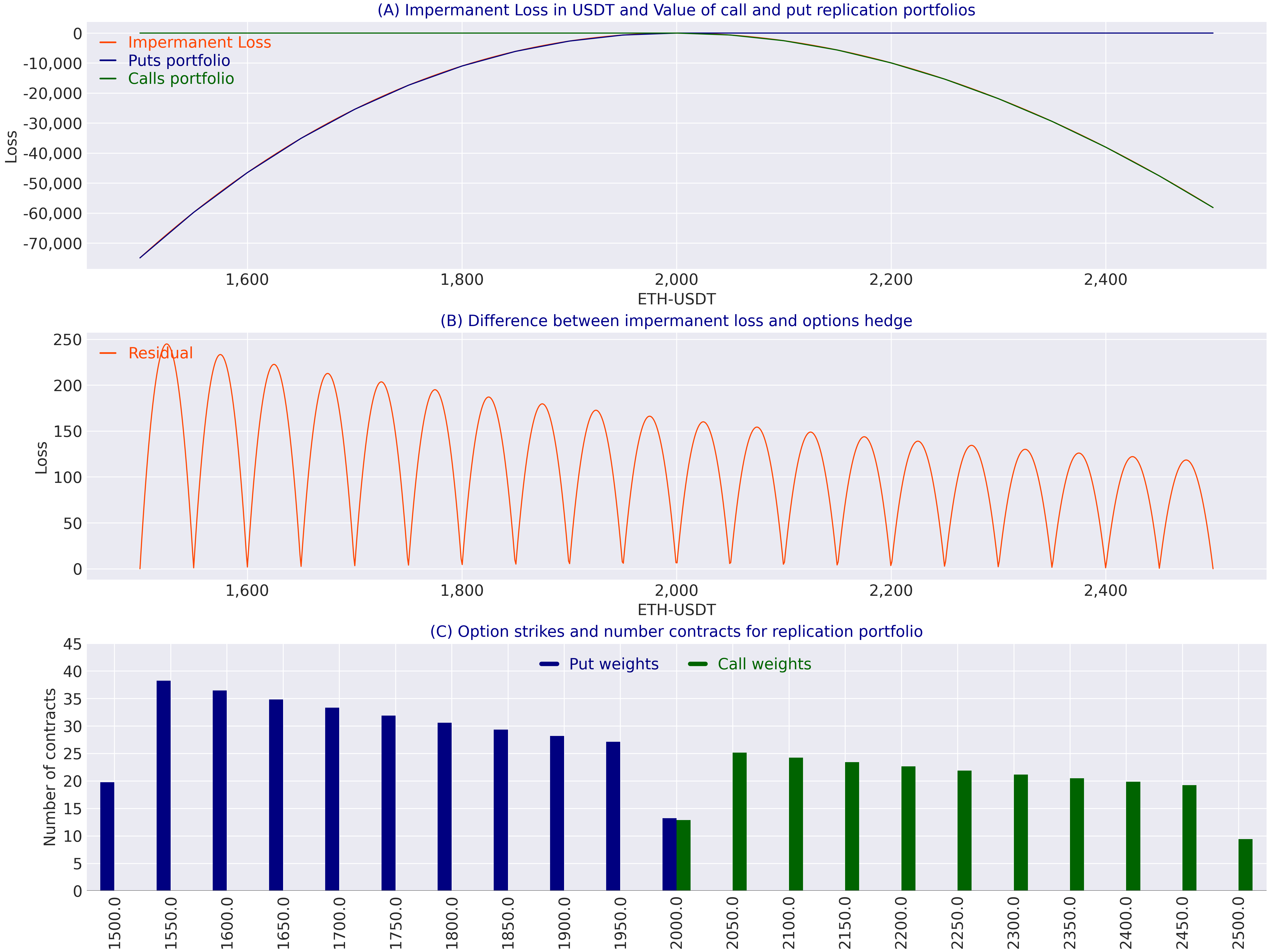}\vspace*{-1.\baselineskip}
\end{center}
\caption{Replication of IL of borrowed Uniswap V3 LP for allocation of $1m$
USDT notional, $p_{0}=2000$ ETH/USDT with $p_{a}=1500$ and $p_{b}=2500$. (A)
Impermanent loss in USDT and (negative) values of replicating puts and call
portfolios; (B) Residual, which is the spread between IL and options
replication portfolios; (C) Number of option contracts for put and calls
portfolios.}
\label{il_replication_borrowed}
\end{figure}

\begin{figure}[H]
\begin{center}
\includegraphics[width=0.8\textwidth, angle=0]
{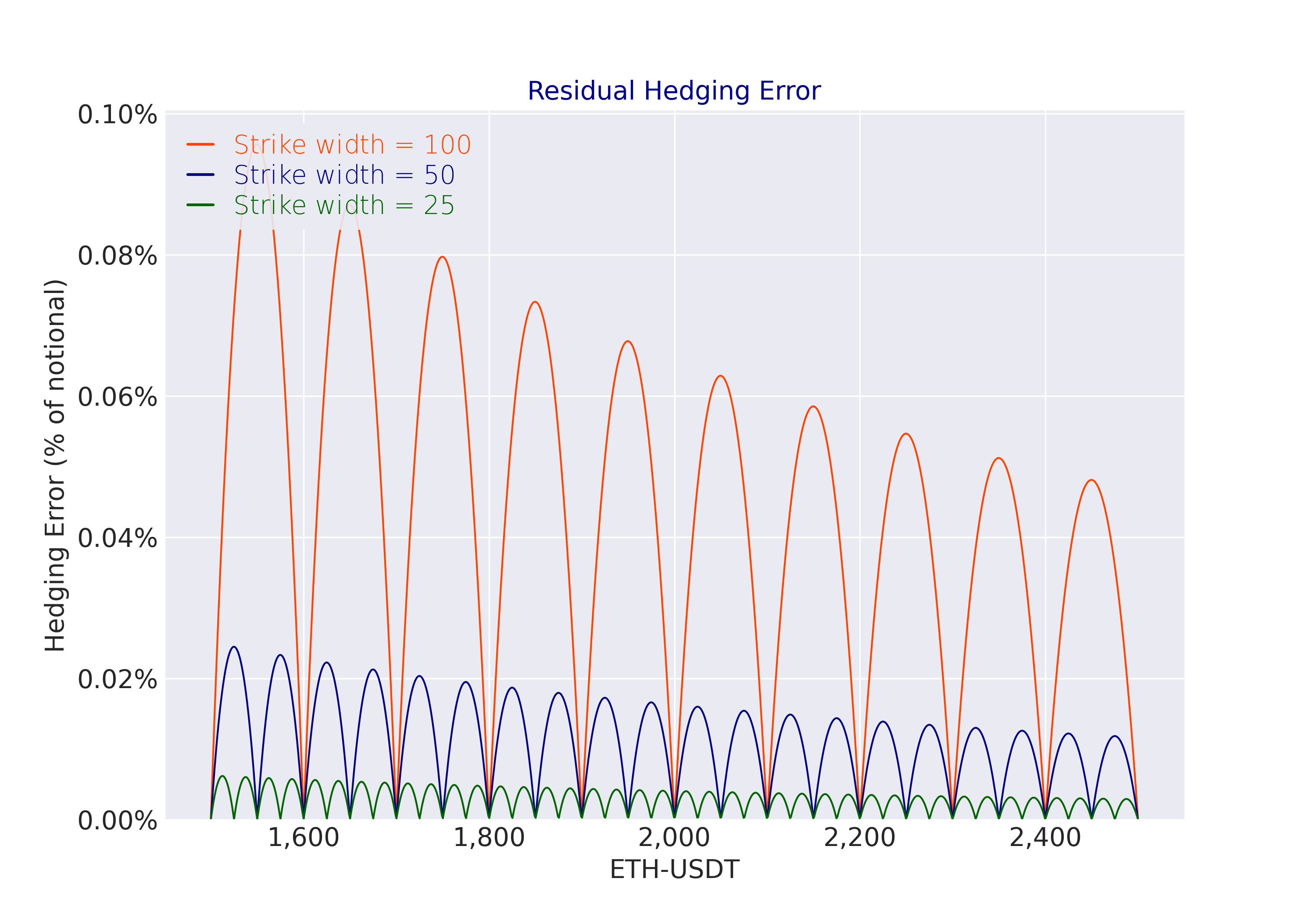}\vspace*{-1.\baselineskip}
\end{center}
\caption{Residual hedging error (as \% of LP notional) arising using options on the discrete strikes with the widths equal to $100$, $50$, and $25$ arising from option-based replication of IL of borrowed Uniswap V3 LP for allocation of $1m$ USDT notional, $p_{0}=2000$ ETH/USDT with $p_{a}=1500$ and $p_{b}=2500$.}
\label{strike_hedging_error}
\end{figure}

\section{Model-dependent Valuation of Protection Claims against IL}

\label{sec:mdp}

In this section, we develop the model-dependent valuation and dynamic
hedging of IL protection claims for Uniswap V2 and V3 protocols.

\subsection{Exponential Price Dynamics}

We consider a continuous-time market with a fixed horizon date $T^{*}>0$ and
uncertainty modelled on the probability space $(\Omega, \mathbb{F}, \mathbb{P})$
equipped with filtration $\mathbb{F}=\{\mathcal{F}_t\}_{0\le t\le T^*}$. We
assume that $\mathbb{F}$ is right-continuous and satisfies the usual conditions.

We introduce exponential price dynamics for price $p_{t}$: 
\begin{equation}  \label{eq:mgf1}
p_{t} = p_{0}e^{x_{t}}, \ x_{0} = 0,
\end{equation}
where $x_{t}$ is a stochastic process driving the log performance. We assume the existence of
a risk-neutral measure $\mathbb{Q}$ such that%
\footnote{%
This assumption is valid for complete markets. For incomplete markets, e.g.
dynamics including stochastic volatility or jumps, we fix a martingale
measure using specific risk preferences (see for example \citet{Lewis2000}%
). \citet{SeppRakhmonoV2023} consider the existence of equivalent
risk-neutral measures for stochastic volatility models.}: 
\begin{equation}  \label{eq:mgf1a}
\begin{split}
\mathbb{E}^{\mathbb{Q}}[p_{T}\left.\right|\mathcal{F}_{t}] & \equiv p_{t}%
\mathbb{E}^{\mathbb{Q}}[e^{x_{T}}\left.\right|\mathcal{F}_{t}] = p_{t}
e^{(r-q)(T-t)}
\end{split}%
\end{equation}
where $r$ is the discount rate and $q$ is the borrow rate.

For hedging an IL protection claim, a trader needs to sell short the
underlying token. Short-selling can be readily executed either through a
centralised exchange (CEX) using perpetual futures or through borrowing the
token on a DeFi protocol using stablecoins as collateral. In the CEX case,
the borrowing rate $q$ is the negative of the funding rate reported by CEX (by
the convention of crypto CEXes, the funding rate is the rate paid by traders
with long positions). In the DeFi case, the borrowing rate $q$ is the accrued
borrow rate.

Regarding the discount rate $r$, most DEXes and CEXes (such as Deribit when
marking their listed options) assume a zero discount rate. We would call $r$
the low risk opportunity cost available in DeFi with the risk being a
potential hack of blockchain technology when deposited and staked assets
could be appropriated. Staking of high quality stablecoins in top DeFi
protocols would yield $1\%-2\%$ in the current environment, which is far
less than rates on government short-term bonds in traditional markets ($%
4\%-5\%$ as of June 2024).

When we value issued IL protection claims, we emphasize that the entry price 
$p_{0}$ is fixed at the time of initialising an LP so that the entry
price $p_{0}$ becomes a parameter of the IL formula in Eq. %
\eqref{eq:il_dec_borrowed1} along with the lower and upper ranges. Once the
valuation time has advanced to time $t$ and the new price is observed, we need to
compute the expected IL using the price $p_{t}$. We introduce the following
decomposition at time $t$ for the total log-performance $x_{T}$ 
\begin{equation}  \label{eq:mgf1b}
x_{T} = x_{t} + x_{\tau},
\end{equation}
where $x_{t}$ is the realised log-performance over period $(0, t]$ and $%
x_{\tau}$ is the stochastic performance over the period $(t, T]$. Here $\tau$%
, $\tau=T-t$, is the time to maturity. As a result, we model $p_{T}$ using
Eq. \eqref{eq:mgf1} with \eqref{eq:mgf1b} as follows 
\begin{equation}  \label{eq:mgf1bb}
p_{T}=p_{0}e^{x_{t} + x_{\tau} }.
\end{equation}

\subsection{Model-dependent Valuation}

We consider the valuation of claims for IL protection as defined in Eqs. \eqref{eq:il_funded_p} and \eqref{eq:il_borrowed_p}. We focus on the
valuation of these pay-offs for Uniswap V2 and V3 using the pay-off
decomposition and by applying the exponential model in Eqs \eqref{eq:mgf1}
and \eqref{eq:mgf1bb}. Here, we fix the maturity time $T$.

\begin{corollary}[Representation of Payoff of IL protection claim in Uniswap V2 AMM under exponential model (\protect\ref{eq:mgf1})]
Applying the exponential model in Eq. \eqref{eq:mgf1bb} to Eqs %
\eqref{eq:il_payoff_v2}, we obtain 
\begin{equation}  \label{eq:il_payoff_v2b}
\begin{split}
& Payoff^{funded}(x_{T}) = 1 - e^{ \frac{1}{2} (x_{t} + x_{\tau})} , \\
& Payoff^{borrowed}(x_{T}) = \frac{1}{2} \left( e^{ \frac{1}{2} (x_{t} +
x_{\tau})} - 1\right)^{2}.
\end{split}%
\end{equation}
\end{corollary}

\begin{corollary}[Representation of Payoff of claim for IL protection for Uniswap V3 AMM under exponential model (\protect\ref{eq:mgf1})]
Applying the exponential model in Eq. \eqref{eq:mgf1bb} to Eqs %
\eqref{eq:il_payoff3}, we obtain 
\begin{equation}  \label{eq:il_payoff_e1}
\begin{split}
& \mathrm{Payoff \ funded}(x_{\tau}) = - \frac{1}{\bar{V}_{0}}\left[u^{funded}_{0}(x_{t}+x_{\tau}) +
2u_{1/2}(x_{t}+x_{\tau}) + u_{1}(x_{t}+x_{\tau})\right], \\
& \mathrm{Payoff \ borrowed} (x_{\tau}) = - \frac{1}{\bar{V}_{0}}\left[u^{borrowed}_{0}(x_{t}+x_{\tau}) +
2u_{1/2}(x_{t}+x_{\tau}) + u_{1}(x_{t}+x_{\tau})\right],
\end{split}%
\end{equation}
where the linear part is computed by 
\begin{equation}  \label{eq:lp1}
\begin{split}
& u^{funded}_{0}(x_{t}+x_{\tau}) = -\frac{p_{0}}{\sqrt{p_{b}}}
e^{x_{t}+x_{\tau}} + \left(\frac{p_{0}}{\sqrt{p_{b}}}-2\sqrt{p_{0}}\right) \\
& u^{borrowed}_{0}(x_{t}+x_{\tau}) =- \sqrt{p_{0}} \left(e^{x_{t}+x_{\tau}}
+1\right), \\
\end{split}%
\end{equation}
the square root part is computed by: 
\begin{equation}  \label{eq:mgf3}
\begin{split}
u_{1/2}(x_{t}+x_{\tau}) & = \sqrt{p_{0}}\exp\left\{\frac{1}{2}
(x_{t}+x_{\tau}) \right\} \mathbbm{1}\left\{ x_{a} < x_{t}+x_{\tau} <
x_{b}\right\}
\end{split}%
\end{equation}
with $x_{a}=\ln(p_{a}/p_{0})$ and $x_{b}=\ln(p_{b}/p_{0})$.

The option part is computed by 
\begin{equation}  \label{eq:mgf2}
\begin{split}
u_{1}(x_{t}+x_{\tau}) = & -\frac{1}{\sqrt{p_{a}}}\max\left\{
p_{a}-p_{0}e^{x_{t}+x_{\tau}}, 0\right\} +\frac{1}{\sqrt{p_{b}}} \max\left\{
p_{0}e^{x_{t}+x_{\tau}}-p_{b}, 0 \right\} \\
& +2\sqrt{p_{a}} \mathbbm{1}\left\{x_{t}+x_{\tau}\leq x_{a}\right\} +2\sqrt{%
p_{b}}\mathbbm{1}\left\{x_{t}+x_{\tau} \geq x_{b}\right\}.
\end{split}%
\end{equation}
\end{corollary}

Given pay-off function in Eqs \eqref{eq:il_payoff_v2b} and %
\eqref{eq:il_payoff_e1}, the present value of the IL protection claim at
time $t$ is computed under the risk-neutral measure $\mathbb{Q}$ by: 
\begin{equation}  \label{eq:mgf4}
\begin{split}
& PV(t,p_{t};T) = e^{-r \tau}\mathbb{E}^{\mathbb{Q}}[Payoff
(x_{\tau})\left.\right|\mathcal{F}_{t}].
\end{split}%
\end{equation}

\subsection{Valuation in Black-Scholes-Merton (BSM) model}

We consider the BSM model with the price dynamics under the risk-neutral
measure $\mathbb{Q}$ given by 
\begin{equation}  \label{eq:bs1}
dp_{t} = \mu p_{t} dt + \sigma p_{t} dw_{t}, \ p_0 = p,
\end{equation}
where $\mu=r-q$ is the risk-neutral drift, $w_{t}$ is a Brownian motion with 
$w_{0}=0$. Accordingly, the log-performance $x_t=\log p_t / p_0$ is driven
by 
\begin{equation}  \label{eq:bs2}
dx_{t} = \left(\mu-\frac{1}{2}\sigma^{2}\right) dt + \sigma dw_{t}, \ x_0 = 0,
\end{equation}
and the distribution of $x_{\tau}$ is normal with mean $\left(\mu-\frac{1}{2}%
\sigma^{2}\right) \tau$ and volatility $\sigma \sqrt{\tau}$.

We note that the distribution of returns of cryptocurrencies deviates from normality (see, for example, Chapter 10 in \citet{LiptonTreccani2021}). We propose to use BSM model as a simple quotation tool for marking the value of IL protection claim using log-normal volatility as the key input parameter. We also note that \citet{Alexander2023a} provide a comprehensive analysis that BSM model with correction for the implied volatility skew can be applied for delta-hedging of options on cryptocurrencies. Thus, we can also use BSM model with a corresponding adjustment for delta-hedging of IL protection claim.

\begin{proposition}[BSM Value of IL protection claim under Uniswap V2]
\label{prop:bs_v2} Applying model dynamics \eqref{eq:bs1} to pay-off
functions in Eq. \eqref{eq:il_payoff_v2b} using valuation operator in Eq. %
\eqref{eq:mgf4}, we obtain 
\begin{equation}  \label{eq:il_payoff_v2bb}
\begin{split}
& PV^{funded}(t,p_{t}) = e^{-r \tau} \left[ 1 - e^{ \frac{1}{2} x_{t}} G
\left(\tau; -\frac{1}{2}\right) \right], \\
& PV^{borrowed}(t,p_{t}) = \frac{1}{2} e^{-r \tau} \left[ e^{(x_{t}
+\mu\tau)} - 2e^{ \frac{1}{2} x_{t}} G \left(\tau; -\frac{1}{2}\right) + 1%
\right],
\end{split}%
\end{equation}
where 
\begin{equation}  \label{eq:bs2a}
G\left(\tau; -\frac{1}{2}\right) = \exp \left\{ \frac{1}{2} \left(\mu-\frac{1%
}{2}\sigma^{2}\right) \tau + \frac{1}{8} \sigma^2 \tau \right\}.
\end{equation}
\end{proposition}

\begin{proof}
For the dynamics in Eq. \eqref{eq:bs2}, we apply the fact that the moment generation function $G(\tau; \phi)$ of $x_{\tau}$ is given by:
\begin{equation} \label{eq:bs2gmf}
G(\tau; \phi) \equiv \mathbb{E}^{\mathbb{Q}}[\exp \left\{-\phi x_{\tau} \right\}\left.\right|\mathcal{F}_{t}] =  \exp \left\{-\left(\mu-\frac{1}{2}\sigma^{2}\right) \phi \tau + \frac{1}{2} \phi^2\sigma^2 \tau  \right\}.
\end{equation}
\end{proof}

\begin{proposition}[BSM value of IL protection claim in Uniswap V3]
\label{prop:bsf}

The values of IL protection claims with pay-off functions in Eq. %
\eqref{eq:il_payoff_e1} under valuation operator in Eq. (\ref{eq:mgf4}) and
BSM dynamics \eqref{eq:bs2} are given by 
\begin{equation}  \label{eq:bsf1}
\begin{split}
& PV^{funded}(t, p_{t}) = - \frac{1}{\bar{V}_{0}}\left[U^{funded}_{0}(p_{t}) + 2U_{1/2}(p_{t}) +
U_{1}(p_{t})\right], \\
& PV^{borrowed}(t, p_{t}) = - \frac{1}{\bar{V}_{0}}\left[U^{borrowed}_{0}(p_{t}) + 2U_{1/2}(p_{t})
+ U_{1}(p_{t})\right],
\end{split}%
\end{equation}
where $U_{0}$, $U_{1/2}$, and $U_{1}$ are the present values of the pay-offs $u_{0}$, $u_{1/2}$, and $u_{1}$ in Eqs. \eqref{eq:lp1}--\eqref{eq:mgf2}, and $\bar{V}_{0}$ is given in Eq. \eqref{eq:V3p}. The linear part is computed by: 
\begin{equation}  \label{eq:bsf2}
\begin{split}
& U^{funded}_{0}(t, p_{t}) = -e^{-q\tau}\frac{p_{t}}{\sqrt{p_{b}}} +
e^{-r\tau}\left(\frac{p_{0}}{\sqrt{p_{b}}}-2\sqrt{p_{0}}\right) \\
& U^{borrowed}_{0}(t, p_{t}) = -\sqrt{p_{0}}\left(\frac{p_{t}}{p_{0}}%
e^{-q\tau} + e^{-r\tau}\right).
\end{split}%
\end{equation}

Here $U_{1/2}(p_{t})$ is the BSM value of the square-root pay-off in Eq. (\ref%
{eq:mgf3}) computed by: 
\begin{equation}
\begin{split}
U_{1/2}(p_t)={}&e^{-r\tau}\sqrt{p_t}
\exp\left\{\frac12\mu\tau-\frac18\sigma^2\tau\right\}\\
&\times\left[\mathbf N\left(\frac{\ln(p_b/p_t)-(r-q)\tau}{\sigma\sqrt\tau}\right)
-\mathbf N\left(\frac{\ln(p_a/p_t)-(r-q)\tau}{\sigma\sqrt\tau}\right)\right].
\end{split}
\end{equation}
where $\mathbf{N}$ is the cpdf of normal random variable.

The option part is computed by 
\begin{equation}  \label{eq:bsf3}
\begin{split}
U_{1}(p_{t}) = & -\frac{1}{\sqrt{p_{a}}} O^{BSM}(p_{t}; p_{a}, -1) +\frac{1}{%
\sqrt{p_{b}}} O^{BSM}(p_{t}; p_{b}, +1) \\
& +2\sqrt{p_{a}} D^{BSM}(p_{t}; p_{a}, -1) +2\sqrt{p_{b}}D^{BSM}(p_{t};
p_{b}, +1).
\end{split}%
\end{equation}

Here $O^{BSM}(p_{t}; k, +1)$ and $O^{BSM}(p_{t}; k,-1)$ are the BSM prices
of vanilla call and put with the strike price $k$ and indicator $%
\omega\in\{+1,-1\}$ respectively: 
\begin{equation}  \label{eq:bsm}
O^{BSM}(p_{t}; k, \omega) = \omega\left[e^{-q\tau} p_{t} \mathbf{N} (\omega d_{+}(p_{t},
k)) - ke^{-r\tau} \mathbf{N} (\omega d_{-}(p_{t}, k))\right],
\end{equation}
where 
\begin{equation}  \label{eq:d2}
d_{\pm}(p_{t},k) = \frac{\ln(p_{t}/k)+(r-q)\tau}{\sigma\sqrt{\tau}}\pm\frac{1%
}{2}\sigma\sqrt{\tau},
\end{equation}
and $D^{BSM}(p_{t}; k,+1)$ and $D^{BSM}(p_{t}; k, -1)$ are digital call and
put options with strike price $k$, respectively: 
\begin{equation}  \label{eq:dg}
D^{BSM}(p_{t}; k, \omega) = e^{-r \tau} \mathbf{N} (\omega d_{-}(p_{t}, k)).
\end{equation}
\end{proposition}

\begin{proof}
See Appendix \ref{ap_prop:bsf}.
\end{proof}

In Figure \ref{fig_bsm_value_in_range}, we show the annualised cost (APR) $%
\%$ for the cost of BSM hedge for the borrowed LP as a function of the range
multiple $m$ such that $p_{a}(m)=e^{-m}p_{0}$ and $p_{b}(m)=e^{m}p_{0}$. We
use two weeks to maturity $T=14/365$ and different values of log-normal
volatility $\sigma$. All else being equal, it is more expensive to hedge
narrow ranges.

\begin{figure}[]
\begin{center}
\includegraphics[width=0.7\textwidth, angle=0]
{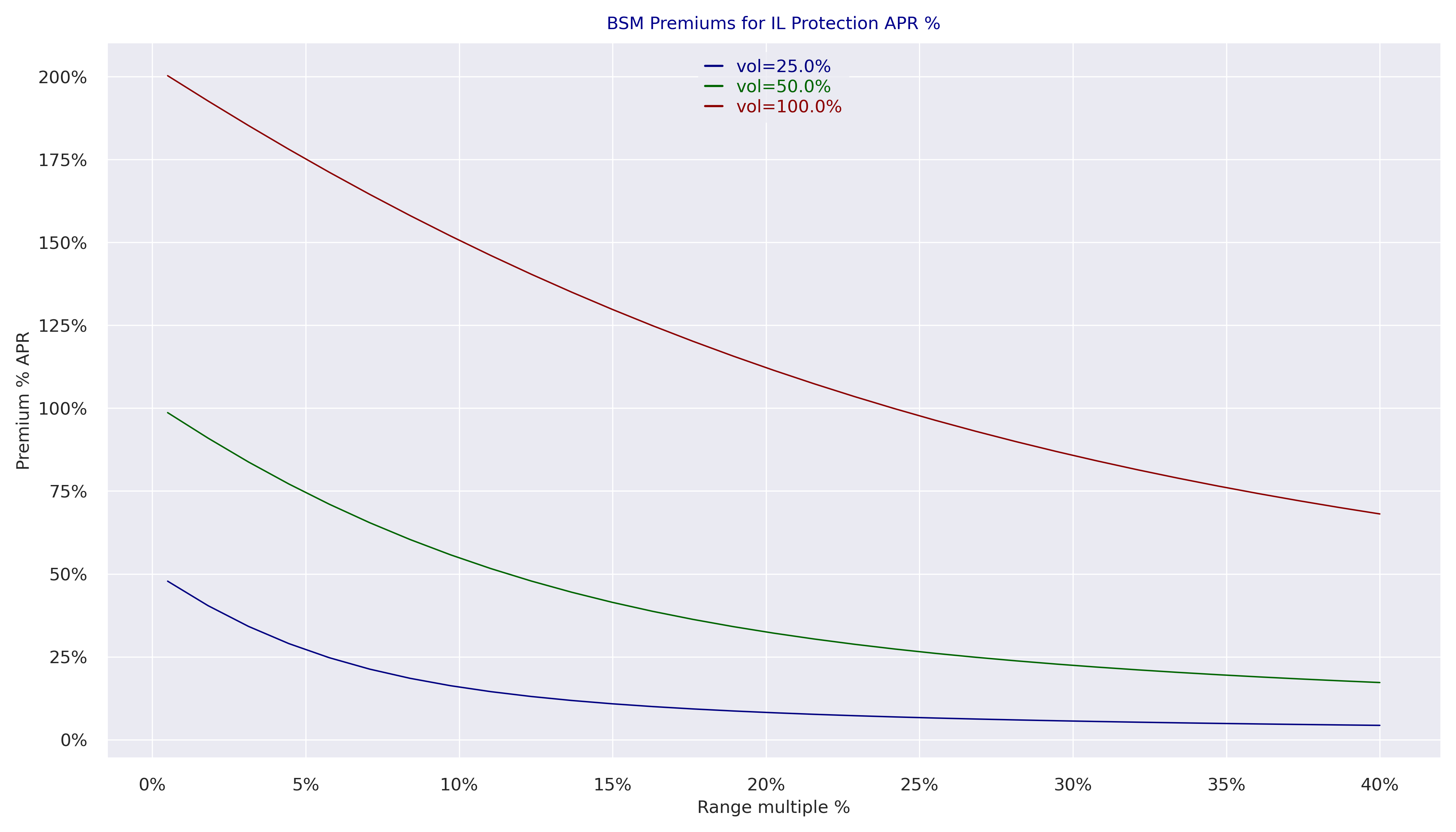}\vspace*{-1.\baselineskip}
\end{center}
\caption{BSM premium annualised ($PV^{borrowed}(t, p_{t})/T$) for borrowed LP
computed using Eq. \eqref{eq:bsf1} with time to maturity of two weeks $%
T=14/365$ and notional of $1$ USDT as function of the range multiple $m$
such that $p_{a}(m)=e^{-m}p_{0}$ and $p_{b}(m)=e^{m}p_{0}$.}
\label{fig_bsm_value_in_range}
\end{figure}

The protection seller must hedge the claim dynamically using option delta.
As a result, we consider the option delta under BSM model.

\begin{corollary}[The delta of IL protection claim in BSM model]
\begin{equation}
\begin{split}
& \Delta ^{funded}(t,p_{t})=-\frac{1}{\bar{V}_{0}}\left[ \partial
_{p}U_{0}^{funded}(p_{t})+2\partial _{p}U_{1/2}(p_{t})+\partial
_{p}U_{1}(p_{t})\right] , \\
& \Delta ^{borrowed}(t,p_{t})=-\frac{1}{\bar{V}_{0}}\left[ \partial
_{p}U_{0}^{borrowed}(p_{t})+2\partial _{p}U_{1/2}(p_{t})+\partial
_{p}U_{1}(p_{t})\right] ,
\end{split}
\label{eq:bs7a}
\end{equation}%
where 
\begin{equation}
\partial _{p}U_{0}^{funded}(t,p_{t})=-e^{-q\tau }\frac{1}{\sqrt{p_{b}}},\
\partial _{p}U_{0}^{borrowed}(t,p_{t})=-e^{-q\tau }\frac{1}{\sqrt{p_{0}}},
\label{eq:bs7b}
\end{equation}%
and 
\begin{equation}
\begin{split}
& \partial _{p}U_{1/2}(p_{t})=\frac{e^{-r\tau }}{2\sqrt{p_{t}}}\exp \left\{ \frac{1}{2}%
\mu \tau -\frac{1}{8}\sigma ^{2}\tau \right\} \left( \mathbf{N}\left( \frac{%
\ln (p_{b}/p_{t})-(r-q)\tau }{\sigma \sqrt{\tau }}\right) -\mathbf{N}\left( 
\frac{\ln (p_{a}/p_{t})-(r-q)\tau }{\sigma \sqrt{\tau }}\right) \right) \\
& -\frac{e^{-r\tau }}{\sigma \sqrt{\tau }\sqrt{p_{t}}}\exp \left\{ \frac{1}{2}\mu \tau
-\frac{1}{8}\sigma ^{2}\tau \right\} \left( \mathbf{n}\left( \frac{\ln
(p_{b}/p_{t})-(r-q)\tau }{\sigma \sqrt{\tau }}\right) -\mathbf{n}\left( 
\frac{\ln (p_{a}/p_{t})-(r-q)\tau }{\sigma \sqrt{\tau }}\right) \right) .
\end{split}
\label{eq:bs8}
\end{equation}%
Finally 
\begin{equation}
\begin{split}
\partial _{p}U_{1}(p_{t})& =-\frac{1}{\sqrt{p_{a}}}\Delta
_{O}(p_{t};p_{a},-1)+\frac{1}{\sqrt{p_{b}}}\Delta _{O}(p_{t};p_{b},+1) \\
& +2\sqrt{p_{a}}\Delta _{D}(p_{t};p_{a},-1)+2\sqrt{p_{b}}\Delta
_{D}(p_{t};p_{b},+1)
\end{split}
\label{eq:bs7}
\end{equation}%
where $\Delta _{O}(p_{t};k,\omega)$ and $\Delta _{D}(p_{t};k,\omega)$ are
Black-Scholes-Merton deltas of vanilla option and digital option,
respectively, given by: 
\begin{equation}
\Delta _{O}(p_{t};k,\omega )=\omega e^{-q\tau }\mathbf{N}(\omega d_{+}(p_{t},k)),\ \Delta
_{D}(p_{t};k,\omega )=\frac{\omega e^{-r\tau }}{p_{t}\sigma \sqrt{\tau }}%
\mathbf{n}(d_{-}(p_{t},k)).  \label{eq:bs8a}
\end{equation}
\end{corollary}

\begin{proof}
Taking the partial derivative wrt $p$ in Eq. (\ref{eq:bsf1}).
\end{proof}

In Figure \ref{bsm_value_in_range}, we show the BSM delta for the borrowed LP calculated using Eq. \eqref{eq:bs7a} with time to maturity of two weeks $%
T=14/365$ for borrowed LP in Uniswap V3 with $p_{0}=2000$, $p_{a}=1500$, $%
p_{b}=2500$, and notional of $1000000$ USD. The initial units in the LP are $%
(220.36,559282.18)$ units of ETH and USDT, respectively. The static hedge is
constructed by shorting $220.36$ units of ETH. The first line labelled LP
ETH units corresponds to the excess ETH units of the borrowed LP with zero units
at $p=p_{0}$ and being underhedged on the downside and overhedged on the
upside. ETH option hedge shows the BSM delta computed using Eq. %
\eqref{eq:bs7a} (the hedge for borrowed LP is implemented using the negative
sign of BSM delta). For high volatilities or large maturity times, BSM delta
under-hedges near the range.

\begin{figure}[]
\begin{center}
\includegraphics[width=0.7\textwidth, angle=0] {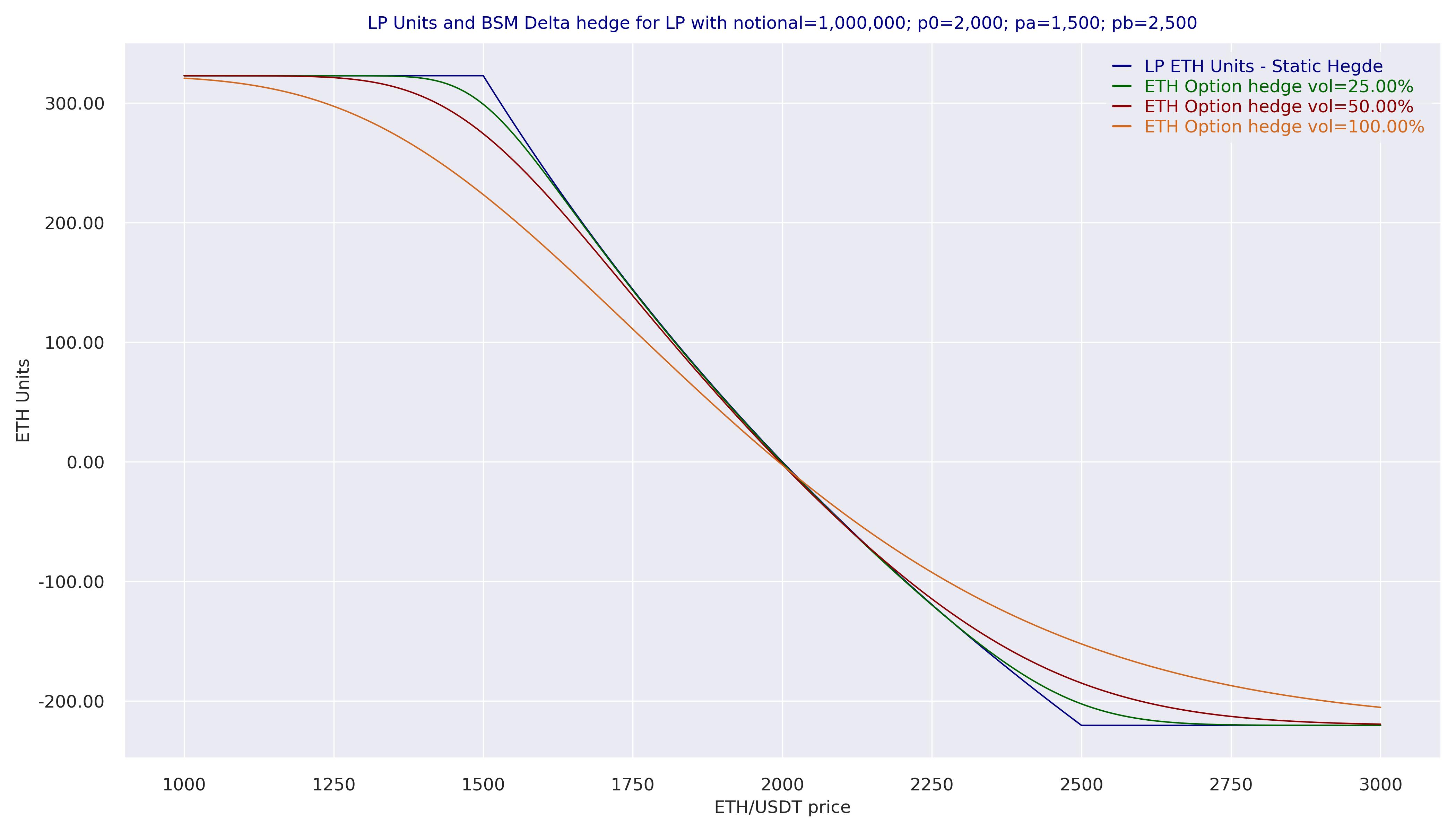}%
\vspace*{-1.\baselineskip}
\end{center}
\caption{BSM delta for borrowed LP computed using Eq. \eqref{eq:bs7a} with
time to maturity of two weeks $T=14/365$ for Uniswap V3 LP with $p_{0}=2000$%
, $p_{a}=1500$, $p_{b}=2500$, and notional of $1000000$ USDT.}
\label{bsm_value_in_range}
\end{figure}

\subsection{Valuation using Moment Generating Function}

We consider a wide class of Markovian exponential dynamics in Eq. \eqref{eq:mgf1} for which the moment generating function (MGF) for
log-return $x_{\tau}$ is available in closed form. The closed-form solution
for the MGF is available under a wide class of models including
jump-diffusions and diffusions with stochastic volatility. Thus, we can
develop an analytic solution for model-dependent valuation of IL protection claims
under various models with analytic MGF.

In particular, \citet{Matic2023} analyse the calibration and hedging of cryptocurrency options under the Heston stochastic volatility model with simultaneous jumps in returns and variance processes and find strong evidence for the presence of stochastic volatility and jumps in the dynamics of implied volatilities of cryptocurrencies. Our MGF-based solution allows us to evaluate and hedge IL protection claims under such models with stochastic volatility and jumps.

We denote the MGF by $G(\tau; \phi)$, where $\phi$ is a complex-valued
transform variable such that $\phi=\phi_{r}+i\phi_{i}$, $i=\sqrt{-1}$.
Formally, the MGF solves the following problem: 
\begin{equation}  \label{eq:mgf9}
G(\tau; \phi) = \mathbb{E}^{\mathbb{Q}}\left[e^{-\phi x_{\tau}}\left.\right|%
\mathcal{F}_{t}\right],
\end{equation}
where the expectation is computed using model dynamics under the
risk-neutral measure $\mathbb{Q}$ and $\mathcal{F}_{t}$ is information set.

For the Fourier pricing formulas below, we distinguish this log-return MGF from the forward-adjusted log-return MGF. Define
\begin{equation*}
F_{t,T}=p_t e^{\mu\tau},\qquad z_{\tau}=\ln(p_T/F_{t,T})=x_{\tau}-\mu\tau,\qquad
E(\tau;\phi)=\mathbb{E}^{\mathbb Q}[e^{-\phi z_{\tau}}\mid\mathcal F_t]
=e^{\phi\mu\tau}G(\tau;\phi).
\end{equation*}
In this subsection's generic payoff transform, set $x=\ln F_{t,T}$ and $x'=\ln p_T$. Thus $x$ is a log-forward coordinate, not the realised log-performance $x_t=\ln(p_t/p_0)$. Assuming a continuous conditional density, the density $P(\tau,x;x')$ of the terminal log-price is given by
\begin{equation}  \label{eq:mgf10}
\begin{split}
P(\tau,x;x')
&=\frac{1}{\pi}\Re\left[\int_0^{\infty}
e^{\phi(x'-\ln p_t)}G(\tau;\phi)\,d\phi_i\right]\\
&=\frac{1}{\pi}\Re\left[\int_0^{\infty}
e^{\phi(x'-x)}E(\tau;\phi)\,d\phi_i\right],
\end{split}
\end{equation}
where $\phi=\phi_r+i\phi_i$, with fixed $\phi_r$ in a common strip of convergence. Below, $d\phi$ denotes $d\phi_i$ along this vertical contour. In particular, the BSM forward-adjusted MGF is $E(\tau;\phi)=\exp\{\tfrac12\sigma^2\tau(\phi^2+\phi)\}$; the carry $\mu\tau$ is included only in $F_{t,T}$.

For the valuation of the IL protection claim using MGF $G(\tau; \phi)$, we need
to evaluate the vanilla and digital options, and the square root pay-off in Eq. \eqref{eq:il_payoff_e1}. First, we derive a generic valuation method for
pay-off $u(\ln p_T)$. We evaluate its present value $U(\tau,x)$ at time
$t$, using $x=\ln F_{t,T}$, which is given similarly to Eq. (\ref%
{eq:mgf4}) by 
\begin{equation}  \label{eq:mgf11}
U(\tau,x) = e^{-r \tau}\mathbb{E}^{\mathbb{Q}}[u(\ln p_T)\left.\right|%
\mathcal{F}_{t}].
\end{equation}

Using Eq.(\ref{eq:mgf10}) we obtain 
\begin{equation}  \label{eq:mgf_pricing}
\begin{split}
U(\tau, x) &=e^{-r\tau} \int^{\infty}_{-\infty}u(x^{\prime })P(\tau,x;
x^{\prime })dx^{\prime } \\
&=e^{-r\tau}\Re\left[ \int^{\infty}_{-\infty}u(x^{\prime })\left[\frac{1}{\pi%
} \int^{\infty}_{0}e^{\phi (x^{\prime }-x)} E(\tau,\phi)d\phi\right]%
dx^{\prime }\right] \\
& = \frac{1}{\pi} e^{-r \tau} \Re\left[\int^{\infty}_{0}\widehat{u}%
(\phi)E(\tau,\phi)d\phi\right],
\end{split}%
\end{equation}
where we assume that the inner integrals are finite to exchange the order of
the integration. Here $\widehat{u}(\phi)$ is the transformed pay-off function
defined by 
\begin{equation}  \label{eq:mgf13}
\widehat{u}(\phi) = e^{-\phi x}\int^{\infty}_{-\infty}e^{\phi x^{\prime }}
u(x^{\prime }) dx^{\prime }.
\end{equation}

\begin{proposition}[Value of IL protection claim in Uniswap V2 under
exponential model dynamics \eqref{eq:mgf1} using the MGF]
\label{prop:bs_v2v} Applying dynamics \eqref{eq:mgf1} to pay-off functions in
Eq. \eqref{eq:il_payoff_v2b} under expectation operator \eqref{eq:mgf4}, we
obtain 
\begin{equation}  \label{eq:il_payoff_v2_mgf}
\begin{split}
& PV^{funded}(t,p_{t}) = e^{-r \tau}\left[ 1 - e^{ \frac{1}{2} x_{t}} G
\left( \tau; -\frac{1}{2}\right) \right], \\
& PV^{borrowed}(t,p_{t}) = \frac{1}{2} e^{-r \tau} \left[ e^{(x_{t}
+\mu\tau)} - 2e^{ \frac{1}{2} x_{t}} G\left(\tau;-\frac{1}{2}\right) + 1%
\right].
\end{split}%
\end{equation}
\end{proposition}

\begin{proof}
We apply the definition of the MGF in Eq. \eqref{eq:mgf9} and the martingale condition in Eq. \eqref{eq:mgf1a}.
\end{proof}Accordingly, the IL under Uniswap V2 can be solved analytically
for a wide class of models which is first concluded in \citet{Lipton2024}.

\begin{proposition}[Valuation of IL protection claim in Uniswap V3 using MGF]
Given the MGF $G(\tau; \phi)$ defined in Eq. \eqref{eq:mgf9} for log-price $%
x_{\tau}$ in the exponential model in Eq. \eqref{eq:mgf1} and pay-off
functions in Eq. \eqref{eq:il_payoff_e1}, we obtain the following valuation
formula 
\begin{equation}  \label{eq:mgfp1}
\begin{split}
& PV^{funded}(t, p_{t}) = - \frac{1}{\bar{V}_{0}}\left[U^{funded}_{0}(p_{t}) + 2U_{1/2}(p_{t}) +
U_{1}(p_{t})\right], \\
& PV^{borrowed}(t, p_{t}) = - \frac{1}{\bar{V}_{0}}\left[U^{borrowed}_{0}(p_{t}) + 2U_{1/2}(p_{t})
+ U_{1}(p_{t})\right],
\end{split}%
\end{equation}
where $U^{funded}_{0}(t, p_{t})$ and $U^{borrowed}_{0}(t, p_{t}) $ are model
independent linear parts computed as in BSM model in Eq. \eqref{eq:bsf2} 
\begin{equation}  \label{eq:mgfp2}
\begin{split}
& U^{funded}_{0}(t, p_{t}) = -e^{-q\tau}\frac{p_{t}}{\sqrt{p_{b}}} +
e^{-r\tau}\left(\frac{p_{0}}{\sqrt{p_{b}}}-2\sqrt{p_{0}}\right) \\
& U^{borrowed}_{0}(t, p_{t}) = -\sqrt{p_{0}}\left(\frac{p_{t}}{p_{0}}%
e^{-q\tau} + e^{-r\tau}\right).
\end{split}%
\end{equation}

$U_{1/2}(p_{t})$ is the square root claim computed using valuation formula %
\eqref{eq:mgf_pricing} with transform in Eq. \eqref{eq:srp3}.

The option part is computed by 
\begin{equation}  \label{eq:mgfp3}
\begin{split}
U_{1}(p_{t}) = & -\frac{1}{\sqrt{p_{a}}} O(p_{t}; p_{a}, -1) +\frac{1}{\sqrt{%
p_{b}}} O(p_{t}; p_{b}, +1) \\
& +2\sqrt{p_{a}} D(p_{t}; p_{a}, -1) +2\sqrt{p_{b}}D(p_{t}; p_{b}, +1),
\end{split}%
\end{equation}
where $O(p_{t}; k, +1)$ and $O(p_{t}; k, -1)$ are model values of call and
put options, respectively, computed in Eq. \eqref{eq:v9}; $D(p_{t};k,+1)$
and $D(p_{t};k,-1)$ are the discounted expectations of the digital pay-offs in Eq. \eqref{eq:do1}, evaluated using the transforms below.
\end{proposition}

We note that for the computation of the pay-off transform $\widehat{u}%
(\phi)$ of capped option in Eq. \eqref{eq:v8} we use the transform variable as $\phi=iy-1/2$, where $y\in(-\infty, + \infty)$. Then for
call and put digitals, we use Eq. \eqref{eq:do2} and \eqref{eq:do4} also set $\phi$ as $%
\phi=iy-1/2$. Finally, the transform of the square root pay-off in Eq. %
\eqref{eq:srp3} can also be evaluated using $\phi=iy-1/2$. Thus a common grid of real-valued transform variable $\{y\}$ with $\phi=iy-1/2$ is sufficient.
We then compute MGF $E(\tau; \phi)$ and Fourier transforms of 5 pay-off functions defined with the transform variable $%
\phi=iy-1/2$, and compute 5 option values using Eq. \eqref{eq:mgf_pricing}.
Accordingly, the numerical implementation of the pricing formula in Eq. %
\eqref{eq:mgfp1} is efficient. The accompanying implementation uses the equally valid contour $\phi=-0.4+iy$, within the common convergence strip, to avoid evaluating the removable square-root singularity at $y=0$ on the $-1/2$ contour.

\begin{proof}

\textbf{Vanilla call and put options}

We represent the put and call pay-off functions using capped pay-offs as
\begin{equation} \label {eq:v11}
\begin{split}
  & c(p_T,k) = (p_T-k)^+=p_T-\min\{p_T,k\},\\
  & p(p_T,k) = (k-p_T)^+=k-\min\{p_T,k\}.
\end{split} 
\end{equation}
Accordingly, we need to evaluate option on the capped pay-off using Eq. (\ref{eq:mgf_pricing}) so that we can value vanilla calls and puts using
\begin{equation} \label{eq:v9}
\begin{split}
  & O(p_{t}; k, +1)  = e^{-q\tau}p_{t}- U(p_{t}, k)\\
  & O(p_{t}; k, -1)  = e^{-r\tau} k - U(p_{t}, k),
\end{split} 
\end{equation}
where $U(p_t,k)$ is the value of the capped pay-off.

Applying Eq. (\ref{eq:mgf13}) for capped pay-off $u(x')=\min \left\{ e^{x'}, k \right\}$, we obtain
\begin{equation} \label {eq:mgf14}
\begin{split}
   \widehat{u}(\phi) & =e^{-\phi x} \left( \frac{e^{(\phi+1) k^{*}}}{\phi+1} - e^{k^{*}}\frac{1}{\phi}e^{\phi k^{*}}\right) = e^{-\phi x}\left(-e^{(\phi+1) k^{*}} \frac{1}{(\phi+1)\phi} \right)= -k e^{-\phi x^{*}} \frac{1}{(\phi+1)\phi},
\end{split} 
\end{equation}
where $x^{*}=x-\ln k=\ln(p_{t}/k)+\mu\tau$ is the forward log-moneyness and $k^{*}=\ln k$ with the first integral being finite for $\Re[\phi]>-1$ and the second integral being finite for $\Re[\phi]<0$. The integral (\ref{eq:mgf14}) is finite for $-1<\phi_{r}<0$. Setting $\phi=iy-1/2$, we derive:
\begin{equation}\label{eq:mgf15}
\widehat{u}(\phi = iy-1/2) = -k e^{-(iy-1/2) x^{*}}  \frac{1}{(1/2+iy)(-1/2+iy)}=
k e^{-(iy-1/2) x^{*}}  \frac{1}{y^{2}+1/4}.
\end{equation}

Finally we obtain the valuation formula for capped pay-off known as Lipton-Lewis formula (\cite{Lipton2001}, \cite{Lewis2000}) as follows
\begin{equation} \label {eq:v8}
\begin{split}
  U(p_{t}, k)&  = \frac{k e^{-r\tau} }{\pi} \Re\left[\int^{\infty}_{0} e^{-(iy-1/2) x^{*}}  \frac{1}{y^{2}+1/4} E(\tau; \phi=iy-1/2)dy\right],
\end{split} 
\end{equation}
where $x^{*}=\ln(p_{t}/k)+\mu\tau$ is log-moneyness.

\textbf{Digital options}

We represent the digital call and put pay-offs as follows:
\begin{equation} \label{eq:do1}
\begin{split}
&d^c(p_T;k)=\mathbbm{1}\{\ln p_T\geq\ln k\},\\
&d^p(p_T;k)=\mathbbm{1}\{\ln p_T\leq\ln k\}.
\end{split}
\end{equation}
Under the continuous conditional law assumed above, equality at the strike has probability zero. Their prices are $D(p_t;k,+1)=e^{-r\tau}\mathbb E^{\mathbb Q}[d^c(p_T;k)\mid\mathcal F_t]$ and similarly for the put.

Writing $k^*=\ln k$ and $x=\ln F_{t,T}$, the digital call transform is
\begin{equation} \label{eq:do2}
\begin{split}
\widehat u^c(\phi)
&=e^{-\phi x}\int_{k^*}^{\infty}e^{\phi x'}\,dx'
=-\frac{e^{-\phi(x-k^*)}}{\phi},
\end{split}
\end{equation}
where the integral converges for $\phi_r<0$.

The digital put transform is
\begin{equation} \label{eq:do3}
\begin{split}
\widehat u^p(\phi)
&=e^{-\phi x}\int_{-\infty}^{k^*}e^{\phi x'}\,dx'
=\frac{e^{-\phi(x-k^*)}}{\phi},
\end{split}
\end{equation}
where the integral converges for $\phi_r>0$. To use a common negative contour for both digitals, apply the price identity
\begin{equation} \label{eq:do4}
D(p_t;k,-1)=e^{-r\tau}-D(p_t;k,+1).
\end{equation}
Thus Eq. \eqref{eq:do2} suffices for both prices with $\phi_r<0$.

\textbf{Square root pay-off}

We evaluate the value function in Eq. (\ref{eq:mgf11}) corresponding to the square root pay-off in Eq. (\ref{eq:mgf3}) as follows
\begin{equation} \label{eq:srp1} 
U_{1/2}(\tau,p_t)=e^{-r\tau}\mathbb E^{\mathbb Q}[u(\ln p_T)\mid\mathcal F_t],
\end{equation}
where
\begin{equation} \label{eq:srp2} 
u(x')=e^{x'/2}\mathbbm{1}\{a<x'<b\},
\end{equation}
where $a=\ln p_a$ and $b=\ln p_b$ are absolute log-price bounds. These are distinct from the entry-relative bounds $x_a,x_b$ in Eq. \eqref{eq:mgf3}.

We compute the transform of the pay-off function in (\ref{eq:mgf_pricing}) as follows
\begin{equation} \label{eq:srp3}
\begin{split}
   \widehat{u}(\phi) & =  e^{-\phi x} \int^{\infty}_{-\infty} e^{\phi x'} u(x') dx'\\
	& = e^{-\phi x}\int^{b}_{a} \exp\left\{\left(\phi+\frac{1}{2}\right)x'\right\} dx'\\
		& = e^{-\phi x}\frac{1}{\left(\phi+\frac{1}{2}\right)} \left[\exp\left\{\left(\phi+\frac{1}{2}\right)b\right\} - \exp\left\{\left(\phi+\frac{1}{2}\right)a\right\} \right],
\end{split} 
\end{equation}
For finite $a,b$, this transform is entire in $\phi$. The apparent singularity at $\phi=-1/2$ is removable, with
\begin{equation*}
\widehat u(-1/2)=e^{x/2}(b-a)=\sqrt{F_{t,T}}\ln(p_b/p_a).
\end{equation*}
This limit is required at $y=0$ if the contour $\phi=-1/2+iy$ is used. If $b\to\infty$ with $a$ finite, convergence requires $\phi_r<-1/2$; if $a\to-\infty$ with $b$ finite, it requires $\phi_r>-1/2$. The unrestricted square-root claim is instead evaluated directly as $e^{-r\tau}\sqrt{p_t}\,G(\tau;-1/2)$.

\end{proof}

\subsection{Application of Log-normal SV Model}

We apply the log-normal SV model which can handle positive correlation
between returns and volatility observed in price-volatility dynamics of
digital assets (see \citet{SeppRakhmonoV2023} for details). We consider
price dynamics under the risk-neutral measure $\mathbb{Q}$ for the spot
price $S_{t}$ and the instantaneous volatility $\sigma_{t}$ as follows 
\begin{equation}  \label{eq:svmd1}
\begin{split}
dS_{t}&=r(t)S_{t}dt+\sigma_{t}S_{t}dw^{(0)}_{t}, \ S_{0} =S, \\
d\sigma_{t}& =\left(\kappa_{1} + \kappa_{2}\sigma_{t} \right)(\theta -
\sigma_{t})dt+ \beta \sigma_{t}dw^{(0)}_{t} + \varepsilon \sigma_{t}
dw^{(1)}_{t}, \ \sigma_{0}=\sigma,
\end{split}%
\end{equation}
where $w^{(0)}$, $w^{(1)}$ are uncorrelated Brownian motions, $\kappa_1>0$
and $\kappa_2 \geq 0$ are linear and quadratic mean-reversion rates
respectively, $\theta>0$ is the mean of the volatility, $\beta\in\mathbb{R}$
is the volatility beta which measures the sensitivity of the volatility to
changes in the spot price, and $\varepsilon>0$ is the volatility of residual
volatility.

For the dynamics in Eq. \eqref{eq:svmd1}, the carry is $\int_t^T r(s)\,ds$ (there is no dividend term). The first-order approximation of \cite{SeppRakhmonoV2023} to the increment MGF in Eq. \eqref{eq:mgf9} is 
\begin{equation}  \label{eq:b11}
G(\tau;\phi)\approx\exp\left\{-\phi\int_t^T r(s)\,ds\right\}E^{[1]}(\tau,\phi),
\end{equation}
where $E^{[1]}$ approximates the forward-adjusted increment MGF $E$. Its coefficients are evaluated at the volatility state $\sigma=\sigma_t$ and it has the exponential form 
\begin{equation}  \label{eq:b12}
\begin{split}
& E^{[1]}(\tau,\phi)= \exp
\left\{\sum^{2}_{k=0}A^{(k)}(\tau;\phi)(\sigma-\theta)^{k} \right\},
\end{split}%
\end{equation}
where $\vartheta^{2}=\beta^{2}+\varepsilon^{2}$ and vector function $\mathbf{%
A}(\tau)=\left\{A^{(k)}(\tau;\phi)\right\}$, $k=0,1,2$, solve the quadratic
differential system as a function of $\tau$: 
\begin{equation}  \label{eq:OrigMGFLogAndQVFirsOrder}
A^{(k)}_\tau=\mathbf{A}^\top M^{(k)}\mathbf{A} + \left(L^{(k)}\right)^\top 
\mathbf{A} + H^{(k)},
\end{equation}
\begin{equation*}
\begin{split}
M^{(k)}&=\left\{ 
\begin{pmatrix}
0 & 0 & 0 \\ 
0 & \frac{\theta^2\vartheta^2}{2} & 0 \\ 
0 & 0 & 0 \\ 
&  & 
\end{pmatrix}%
, 
\begin{pmatrix}
0 & 0 & 0 \\ 
0 & \theta\vartheta^2 & \theta^2\vartheta^2 \\ 
0 & \theta^2\vartheta^2 & 0 \\ 
&  & 
\end{pmatrix}%
, 
\begin{pmatrix}
0 & 0 & 0 \\ 
0 & \frac{\vartheta^2}{2} & 2\theta\vartheta^2 \\ 
0 & 2\theta\vartheta^2 & 2\theta^2\vartheta^2 \\ 
&  & 
\end{pmatrix}
\right\}, \\
L^{(k)}&=\left\{%
\begin{pmatrix}
0 \\ 
-\theta^2\beta\phi \\ 
\theta^2\vartheta^2 \\ 
\end{pmatrix}%
, 
\begin{pmatrix}
0 \\ 
-(\kappa_1+\kappa_2\theta)-2\theta\beta\phi \\ 
2(\theta\vartheta^2-\theta^2\beta\phi) \\ 
\end{pmatrix}%
, 
\begin{pmatrix}
0 \\ 
-\beta\phi-\kappa_2 \\ 
\vartheta^2-2(\kappa_1+\kappa_2\theta)-4\theta\beta\phi%
\end{pmatrix}
\right\}, \\
H^{(k)}&=\left\{ \frac{1}{2}\theta^2\left(\phi ^2+\phi\right), \
\theta\left(\phi ^2+\phi\right), \ \frac{1}{2}\left(\phi
^2+\phi\right) \right\},
\end{split}%
\end{equation*}
with the initial condition $\mathbf{A}(0)=( 0, \ 0, \ 0 )^\top$. The second-order solution is provided in Theorem 4.6 in \cite{SeppRakhmonoV2023}.

In Subplot (A) of Figure \ref{joint_logsv_figure}\footnote{%
Github project \url{https://github.com/ArturSepp/StochVolModels} provides
Python code for these computations using the log-normal SV model.}, we show
the implied volatilities of the log-normal SV model for a range of
volatility of residual volatility $\varepsilon$ with $\beta=0$. In Subplot
(B), we show the premium APR for IL protection as a function of the range multiple for a range of $\varepsilon$. We see that the model-value of IL
protection is not very sensitive to tails of implied distribution (or,
equivalently, to the convexity of the implied volatility). The reason is
that most of the value of IL protection is derived from the center of
returns distribution.

\begin{figure}[]
\begin{center}
\includegraphics[width=0.70\textwidth, angle=0]
{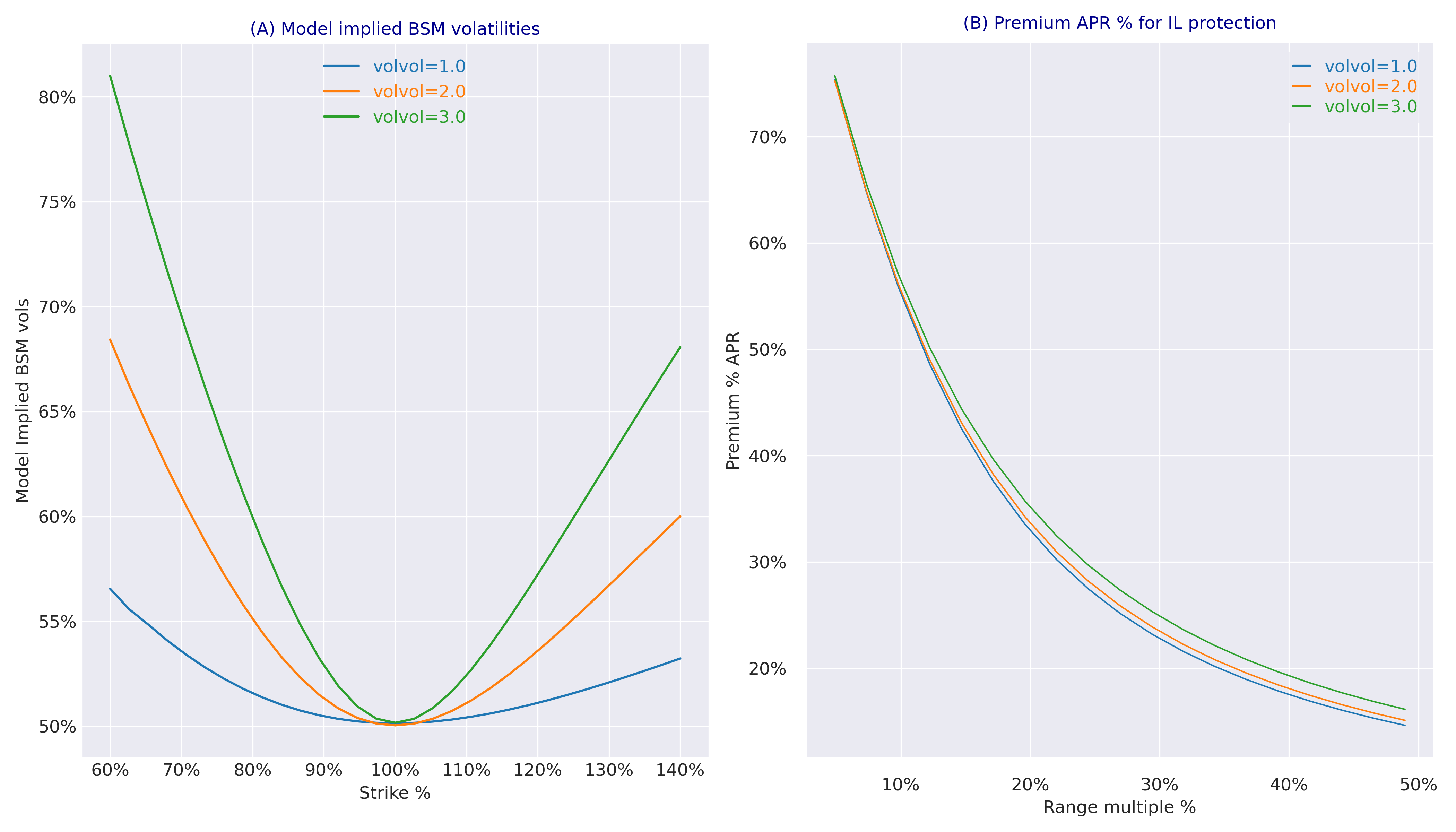}\vspace*{-1.\baselineskip}
\end{center}
\caption{(A) BSM volatilities implied by log-normal SV model as function of
volatility-of-volatility parameter $\protect\varepsilon$ for $\protect\tau%
=14/365$; (B) Premiums APR computed using log-normal SV model for borrowed
LP as function of the range multiple $m$ such that $p_{a}(m)=e^{-m}p_{0}$
and $p_{b}(m)=e^{m}p_{0}$. Other model parameters include $\protect\sigma%
_{0}=\protect\theta=0.50$, $\protect\kappa_{1}=2.21$, $\protect\kappa%
_{2}=2.18$, $\protect\beta=0.0$.}
\label{joint_logsv_figure}
\end{figure}

\section{Practical Considerations}

Hereby we will discuss our experience of implementing the proposed approach for hedging impermanent loss in practice and address some frequently raised questions related to the process.

\textbf{Funded vs borrowed LP}. The funded position is typically chosen by an LP provider with existing holdings of cryptocurrencies. In this case, the funded position is seen as a way to generate yield. The LP provider of the funded position buys the protection to hedge the downside risk. Borrowed LP is typically implemented by an LP provider seeking absolute return, in which case the LP provider borrows an amount of cryptocurrency to stake into the LP and buys the LP protection to hedge against both downside and upside potentials.

\textbf{Hedging and short-selling}. The borrowed LP typically requires two transactions. First, the LP provider purchases the necessary amount of volatile cryptocurrency (Bitcoin, ETH, etc.) on-chain\footnote{For large amounts, it might be cost-efficient to purchase coins on a centralised exchange and to do a transfer to an on-chain wallet.} and simultaneously sells short the same amount through perpetual or termed futures in a centralised crypto exchange (short-selling of most altcoins in moderate amounts can also be implemented in decentralised exchanges). We recommend that the capital efficiency is about $1.3$ meaning that for the $1.0$ ETH staked, the LP provider funds the margin account for the short future position by $0.3$ ETH. Second, the LP provider buys the necessary number of protection options. It is efficient to do delta-hedging with option buying on the same exchange (such as Deribit) because the portfolio margin of the combined option and futures positions is smaller.

\textbf{Hedging using lending and borrowing protocols}. The LP provider can also obtain a short position by using lending protocols. The typical loan-to-value (LTV) ratio is $0.75$, which means that the LP provider must commit the amount of stablecoins worth $1.3$ ETH to borrow $1.0$ ETH. Thus, the capital efficiency is about the same as when short-selling through a centralised exchange. In our experience, borrowing altcoins from a lending platform provides less capacity and it is subject to higher liquidation risk.

\textbf{Latency risk}. LP staking and hedging requires a few transactions for the borrowed LP including (1) buying underlying cryptocurrency on-chain, (2) hedging the exposure by short-selling term futures on a centralised exchange, (3) staking the cryptocurrency to a DEX on-chain, and (4) buying the option protection on a centralised exchange. In our experience, these steps cannot be fully automated because of the exposure to both on-chain and off-chain protocols, different validation and execution times, and very fragmented liquidity of both centralised and decentralised exchanges for relatively modest staking amounts (any position worth above $100$ ETH can be considered as large in the current environment). Consequently, the latency risk for LP staking is present and must be addressed with a robust operational framework.

\textbf{BTC and ETH vs altcoins}. We note that the option-based hedging of IL using vanilla options is currently only possible for BTC and ETH cryptocurrencies. In addition, Deribit exchange currently lists options on XRP, BNB, and SOL, but the liquidity of these options is poor. The IL protection for (top 50) liquid altcoins can be bought from a handful of designated market-makers over-the-counter (OTC) who offer such products.

\textbf{Overall profitability}. Purchasing options is expensive in the long-term, so the LP provider must estimate the volume of a liquidity pool and the contribution of provided liquidity to the liquidity profile of the pool. Ultimately, this decision comes down to the analysis of the risk-reward problem. If the liquidity provider regularly executes LP staking with option-based hedging, the realised P\&L is driven by the difference between realised staking fees and premiums paid for purchasing options for LP IL hedging. We leave a quantitative approach and an empirical analysis of this topic for future research.

\section{Conclusions}

\label{sec:conc}

We have developed a unified approach for hedging of Impermanent Loss (IL)
which arises when providing liquidity to Automated Market Making (AMM) pools
in blockchain ecosystem. We have introduced the two ways to create a
liquidity provision (LP) that includes a funded LP (with the long initial
exposure to the underlying token) and a borrowed LP (with the zero initial
exposure to the underlying token). We have applied Uniswap V2 and V3
protocols, which are the constant function market maker (CFMM) most commonly
employed by most decentralised exchanges. We have shown that the IL can
be represented with a non-linear function of the current spot price. As a
result, using traditional methods of financial engineering, we can handle
the valuation and risk management of the IL protection claim which delivers
the negative of the IL at a fixed maturity date.

First, we have derived a static replication approach for the IL arising from
a generic constant function market maker (CFMM) using a portfolio of traded
call and put options at a fixed maturity date. This approach allows for
model-free replication of the IL when there is a liquid option market, which
is the case for core digital assets, including Bitcoin and Ethereum.

Second, for digital assets without a liquid options market, we have
developed a model-based approach using the decomposition of the IL function
into vanilla options, digital options, and an exotic square root pay-off. We
have derived a closed-form valuation formula for a wide class of price
dynamics with tractable characteristic and moment generating functions (MGF)
by means of Fourier transform.

Model-based valuation can be employed by a few crypto-trading companies that
currently sell over-the-counter IL protection claims. When using model-based
dynamic delta-hedging for the replication of the pay-off of the IL
protection claim, the profit-and-loss (P\&L) of the dynamic delta-hedging
strategy will be primarily driven by the realised variance of the price
process. Thus, the total P\&L of a trading desk will be the difference
between the premiums received (from selling IL protection claims) and the
variance realised through delta hedging. The trading desk can employ our results
for the analysis of price dynamics and hedging strategies that optimise
their total P\&L.

For liquidity providers who buy IL protection claims for their LP position,
the total P\&L will be driven by the difference between the accrued fees from LP
positions and the costs of IL protection claims. The cost of the IL protection
claim can be estimated beforehand using either the cost of static options
replicating portfolio or costs of buying IL protection from a trading desk.
As a result, liquidity providers can focus on selecting DEX pools and
liquidity ranges where expected fees could exceed hedging costs. Thus,
liquidity providers can apply our analysis for optimal allocation to LP pools, for creating static replication portfolios using traded
options, or for assessing costs quoted by providers of IL protection.

We leave the application of our model-free and model-dependent results for an optimal liquidity provision and an optimal design of LP pools for future research.


\section*{Appendices}

\subsection{Proof of Proposition \protect\ref{prop:pnl_funded_V3}}

\label{ap_prop:pnl_funded_V3}

\begin{proof}
Using Eq.(\ref{eq:pnl_funded}) with Eq.(\ref{eq:V34}) for $p_{t} \in (p_{a}, p_{b})$: 
\begin{equation} \label{eq:V36} 
\begin{split}
\mathrm{ P\&L \ funded} & =\left(p_{t} x_{t} + y_{t}\right) - \left(p_{0}x_{0} +  y_{0}\right)\\
& =L \left(p_{t} \left(\frac{1}{\sqrt{p_{t}}} - \frac{1}{\sqrt{p_{b}}}\right) + \left(\sqrt{p_{t}} - \sqrt{p_{a}}\right)\right) - L \left(p_{0} \left(\frac{1}{\sqrt{p_{0}}} - \frac{1}{\sqrt{p_{b}}}\right) + \left(\sqrt{p_{0}} - \sqrt{p_{a}}\right)\right)\\
& =L \left[ \left( \sqrt{p_{t}} - \frac{p_{t}}{\sqrt{p_{b}}}\right) + \left(\sqrt{p_{t}} - \sqrt{p_{a}}\right) -  \left(\sqrt{p_{0}} - \frac{p_{0}}{\sqrt{p_{b}}}\right) - \left(\sqrt{p_{0}} - \sqrt{p_{a}}\right)\right]\\
& =L \left[ \left( \sqrt{p_{t}} - \frac{p_{t}}{\sqrt{p_{b}}}\right) + \sqrt{p_{t}}  -  \left(\sqrt{p_{0}} - \frac{p_{0}}{\sqrt{p_{b}}}\right) - \sqrt{p_{0}}\right]\\
& =L \left[ 2\left(\sqrt{p_{t}} -\sqrt{p_{0}} \right)  +   \frac{p_{0}-p_{t}}{\sqrt{p_{b}}}\right].
\end{split}
\end{equation}

For $p_{t} \leq p_{a}$, using (\ref{eq:V34a}):
\begin{equation} \label{eq:V36a} 
\begin{split}
\mathrm{ P\&L \ funded} & = \left(p_{t} x_{t} + y_{t} \right) - \left(p_{0} x_{0} + y_{0}\right)\\
& = L \left[ \left(p_{t} \left(\frac{1}{\sqrt{p_{a}}} - \frac{1}{\sqrt{p_{b}}}\right) +0\right) - \left(p_{0} \left(\frac{1}{\sqrt{p_{0}}} - \frac{1}{\sqrt{p_{b}}}\right)  + \left(\sqrt{p_{0}} - \sqrt{p_{a}}\right)\right) \right]\\
& = L \left[ p_{t} \left(\frac{1}{\sqrt{p_{a}}} - \frac{1}{\sqrt{p_{b}}}\right) - \left(p_{0} \left(\frac{1}{\sqrt{p_{0}}} - \frac{1}{\sqrt{p_{b}}}\right)  + \left(\sqrt{p_{0}} - \sqrt{p_{a}}\right)\right) \right]\\
& = L \left[ p_{t} \left(\frac{1}{\sqrt{p_{a}}} - \frac{1}{\sqrt{p_{b}}}\right) + \frac{p_{0}}{\sqrt{p_{b}}}  - 2\sqrt{p_{0}} + \sqrt{p_{a}}\right].
\end{split}
\end{equation}

For $p_{t}  \geq  p_{b}$, using Eq. (\ref{eq:V34b}):
\begin{equation} \label{eq:V36b} 
\begin{split}
\mathrm{P\&L \ funded} & = \left(p_{t} x_{t} + y_{t} \right) - \left(p_{0} x_{0} + y_{0}\right)\\
& = L\left[ \left(0 + \left(\sqrt{p_{b}} - \sqrt{p_{a}}\right) \right) - \left(p_{0} \left(\frac{1}{\sqrt{p_{0}}} - \frac{1}{\sqrt{p_{b}}}\right)  + \left(\sqrt{p_{0}} - \sqrt{p_{a}}\right)\right)\right]\\
& = L\left[\left(\sqrt{p_{b}} - \sqrt{p_{a}}\right) - \left(p_{0} \left(\frac{1}{\sqrt{p_{0}}} - \frac{1}{\sqrt{p_{b}}}\right)  + \left(\sqrt{p_{0}} - \sqrt{p_{a}}\right)\right)\right]\\
& = L\left[\left(\sqrt{p_{b}} - \sqrt{p_{a}}\right) + \frac{p_{0}}{\sqrt{p_{b}}}  - 2\sqrt{p_{0}} + \sqrt{p_{a}}\right]\\
& = L\left[ \sqrt{p_{b}}  + \frac{p_{0}}{\sqrt{p_{b}}}  - 2\sqrt{p_{0}} \right].
\end{split}
\end{equation}

\end{proof}

\subsection{Proof to Proposition \protect\ref{prop:pnl_borrowed_V3}}

\label{ap_prop:pnl_borrowed_V3}

\begin{proof} [Proposition \ref{prop:pnl_borrowed_V3}]
Using Eq.(\ref{eq:pnl_borrowed}) with Eq.(\ref{eq:V34}) for $p \in (p_{a}, p_{b})$: 
\begin{equation} \label{eq:V37} 
\begin{split}
\mathrm{P\&L \ borrowed} & = \left(p_{t} x_{t} + y_{t} \right) - \left(p_{t} x_{0} + y_{0}\right)\\
& = L \left[ \left(p_{t} \left(\frac{1}{\sqrt{p_{t}}} - \frac{1}{\sqrt{p_{b}}}\right) + \left(\sqrt{p_{t}} - \sqrt{p_{a}}\right) \right) - \left(p_{t} \left(\frac{1}{\sqrt{p_{0}}} - \frac{1}{\sqrt{p_{b}}}\right)  + \left(\sqrt{p_{0}} - \sqrt{p_{a}}\right)\right) \right]\\
& = L \left[ p_{t} \left(\frac{1}{\sqrt{p_{t}}} - \frac{1}{\sqrt{p_{0}}}\right) + \left(\sqrt{p_{t}} - \sqrt{p_{0}}\right)  \right]\\
& = L \left[ p_{t} \left(\frac{\sqrt{p_{0}}-\sqrt{p_{t}}}{\sqrt{p_{0}}\sqrt{p_{t}}}\right) + \left(\sqrt{p_{t}} - \sqrt{p_{0}}\right)  \right]\\
& = L \left[  \frac{\sqrt{p_{0}p_{t}}-p_{t}}{\sqrt{p_{0}}} + \left(\sqrt{p_{t}} - \sqrt{p_{0}}\right)  \right]\\
& = L \left[  \frac{2\sqrt{p_{0}p_{t}}-p_{t}-p_{0}}{\sqrt{p_{0}}} \right]\\
& =  - \frac{L}{\sqrt{p_{0}}}  \left(\sqrt{p_{t}}-\sqrt{p_{0}}\right)^{2} =  -  L \sqrt{p_{0}} \left(\sqrt{\frac{p_{t}}{p_{0}}}-1\right)^{2}.
\end{split}
\end{equation}

For $p \leq p_{a}$, using (\ref{eq:V34a}):
\begin{equation} \label{eq:V37a} 
\begin{split}
\mathrm{P\&L \ borrowed} & = \left(p_{t} x_{t} + y_{t} \right) - \left(p_{t} x_{0} + y_{0}\right)\\
& = L \left[ \left(p_{t} \left(\frac{1}{\sqrt{p_{a}}} - \frac{1}{\sqrt{p_{b}}}\right) +0\right) - \left(p_{t} \left(\frac{1}{\sqrt{p_{0}}} - \frac{1}{\sqrt{p_{b}}}\right)  + \left(\sqrt{p_{0}} - \sqrt{p_{a}}\right)\right) \right]\\
& = L \left[  p_{t} \left(\frac{1}{\sqrt{p_{a}}} - \frac{1}{\sqrt{p_{0}}}\right)   - \left(\sqrt{p_{0}} - \sqrt{p_{a}}\right) \right].
\end{split}
\end{equation}

For $p  \geq  p_{b}$, using Eq. (\ref{eq:V34b}):
\begin{equation} \label{eq:V37b} 
\begin{split}
\mathrm{ P\&L \ borrowed} & = \left(p_{t} x_{t} + y_{t} \right) - \left(p_{t} x_{0} + y_{0}\right)\\
& = L\left[ \left(0 + \left(\sqrt{p_{b}} - \sqrt{p_{a}}\right) \right) - \left(p_{t} \left(\frac{1}{\sqrt{p_{0}}} - \frac{1}{\sqrt{p_{b}}}\right)  + \left(\sqrt{p_{0}} - \sqrt{p_{a}}\right)\right)\right]\\
& = L\left[   \left(\sqrt{p_{b}} - \sqrt{p_{0}}\right)- p_{t} \left(\frac{1}{\sqrt{p_{0}}} - \frac{1}{\sqrt{p_{b}}}\right) \right].
\end{split}
\end{equation}

\end{proof}

\subsection{Proof to Proposition \protect\ref{prop:il_dec_funded1}}

\label{ap_prop:il_dec_funded1} 
\begin{proof}

We apply P\&L of Funded LP in Eq. (\ref{eq:pnl_funded_V3}) with $L=1$ as follows.

\begin{equation} \label{eq:ilf1} 
\mathrm{P\&L \ funded}^{(y)}(p_{t}) = \begin{cases}
 2\left(\sqrt{p_{t}} -\sqrt{p_{0}} \right)  +   \frac{p_{0}-p_{t}}{\sqrt{p_{b}}}  \quad & \, p_{t} \in (p_{a}, p_{b}) \\
  p_{t} \left(\frac{1}{\sqrt{p_{a}}} - \frac{1}{\sqrt{p_{b}}}\right) + \frac{p_{0}}{\sqrt{p_{b}}}  - 2\sqrt{p_{0}} + \sqrt{p_{a}}  \quad & p_{t} \leq p_{a} \\
\sqrt{p_{b}}  + \frac{p_{0}}{\sqrt{p_{b}}}  - 2\sqrt{p_{0}}\quad & p_{t}  \geq  p_{b}.
     \end{cases}
\end{equation}

We consider the range part in the interval $p_{t} \in (p_{a}, p_{b}) $ as follows.
\begin{equation} \label{eq:ilf1a} 
\begin{split}
\mathrm{Range}(p_{t})&  =  2\left(\sqrt{p_{t}} -\sqrt{p_{0}} \right)  +   \frac{p_{0}-p_{t}}{\sqrt{p_{b}}}  \\
& = - \frac{p_{t}}{\sqrt{p_{b}}} + 2\sqrt{p_{t}} + \left(\frac{p_{0}}{\sqrt{p_{b}}}-2\sqrt{p_{0}}\right).
\end{split}
\end{equation}

We extend the range part to the interval $p_{t}\in\left(0, +\infty\right)$ as follows
\begin{equation} \label{eq:ilf2} 
\begin{split}
\mathrm{Range}(p_{t})&  \equiv - \frac{1}{\sqrt{p_{b}}}u_{1}(p_{t}) + 2u_{1/2}(p_{t}) + \left(\frac{p_{0}}{\sqrt{p_{b}}}-2\sqrt{p_{0}}\right)u_{0}(p_{t}),
\end{split}
\end{equation}
where
\begin{equation} \label{eq:ilf3} 
\begin{split}
& u_{1}(p_{t})  = p_{t} \mathbbm{1}\left\{ p_{a} < p_{t} < p_{b}\right\} \\
& u_{1/2}(p_{t})  =  \sqrt{p_{t}} \mathbbm{1}\left\{ p_{a} < p_{t} < p_{b}\right\}\\ 
& u_{0}(p_{t}) = \mathbbm{1}\left\{ p_{a} < p_{t} < p_{b}\right\}
\end{split}
\end{equation}
are extended on $p_{t}\in\left(0, +\infty\right)$ with zero values outside $p_{t} \in (p_{a}, p_{b}) $.

It is clear that $u_{1}(p_{t})$ can be decomposed on $p_{t}\in(0, \infty)$ as a collar option position along with put and call digital:
\begin{equation} \label{eq:ilf4} 
\begin{split}
u_{1}(p_{t}) & =  p_{t} + \max\left\{ p_{a}-p_{t}, 0\right\} - \max\left\{ p_{t}-p_{b}, 0\right\} - p_{a}\mathbbm{1}\left\{p_{t}\leq p_{a}\right\} - p_{b}\mathbbm{1}\left\{p_{t} \geq p_{b}\right\}.
\end{split}
\end{equation}

We evaluate $u_{0}(p_{t})$ term as follows
\begin{equation} \label{eq:ilf5} 
\begin{split}
u_{0}(p_{t}) & =  1 - \mathbbm{1}\left\{p_{t} \geq p_{b}\right\} - \mathbbm{1}\left\{p_{t}\leq p_{a}\right\}.
\end{split}
\end{equation}

Then $u_{1/2}(p_{t})$ is the only non-linear claim which needs model valuation
\begin{equation} \label{eq:ilf6} 
\begin{split}
u_{1/2}(p_{t}) & =  \sqrt{p_{t}} \mathbbm{1}\left\{ p_{a} < p_{t} < p_{b}\right\}.
\end{split}
\end{equation}

Thus the range part on $p_{t}\in(0, \infty)$ becomes
\begin{equation} \label{eq:ilf7a} 
\begin{split}
& \mathrm{Range}(p_{t})  \equiv - \frac{1}{\sqrt{p_{b}}}u_{1}(p_{t}) + 2u_{1/2}(p_{t}) + \left(\frac{p_{0}}{\sqrt{p_{b}}}-2\sqrt{p_{0}}\right)u_{0}(p_{t})\\
& = - \frac{1}{\sqrt{p_{b}}} \left[p_{t} + \max\left\{ p_{a}-p_{t}, 0\right\} - \max\left\{ p_{t}-p_{b}, 0\right\} - p_{a}\mathbbm{1}\left\{p_{t}\leq p_{a}\right\} - p_{b}\mathbbm{1}\left\{p_{t} \geq p_{b}\right\}\right] + 2u_{1/2}(p_{t})\\
& + \left(\frac{p_{0}}{\sqrt{p_{b}}}-2\sqrt{p_{0}}\right) \left[1 - \mathbbm{1}\left\{p_{t} \geq p_{b}\right\} - \mathbbm{1}\left\{p_{t}\leq p_{a}\right\} \right]\\
& = - \frac{1}{\sqrt{p_{b}}} p_{t} + 2u_{1/2}(p_{t})+ \left(\frac{p_{0}}{\sqrt{p_{b}}}-2\sqrt{p_{0}}\right)  \\
& - \frac{1}{\sqrt{p_{b}}} \max\left\{ p_{a}-p_{t}, 0\right\} + \frac{1}{\sqrt{p_{b}}} \max\left\{ p_{t}-p_{b}, 0\right\} + \frac{1}{\sqrt{p_{b}}}p_{a}\mathbbm{1}\left\{p_{t}\leq p_{a}\right\} + \sqrt{p_{b}}\mathbbm{1}\left\{p_{t} \geq p_{b}\right\}\\
& - \left(\frac{p_{0}}{\sqrt{p_{b}}}-2\sqrt{p_{0}}\right) \left[\mathbbm{1}\left\{p_{t} \geq p_{b}\right\} + \mathbbm{1}\left\{p_{t}\leq p_{a}\right\} \right] \\
& = - \frac{1}{\sqrt{p_{b}}} p_{t} + 2u_{1/2}(p_{t})+ \left(\frac{p_{0}}{\sqrt{p_{b}}}-2\sqrt{p_{0}}\right)  - \frac{1}{\sqrt{p_{b}}} \max\left\{ p_{a}-p_{t}, 0\right\} + \frac{1}{\sqrt{p_{b}}} \max\left\{ p_{t}-p_{b}, 0\right\} \\
& + \left[ \frac{1}{\sqrt{p_{b}}}p_{a}   - \left(\frac{p_{0}}{\sqrt{p_{b}}}-2\sqrt{p_{0}}\right)   \right]\mathbbm{1}\left\{p_{t}\leq p_{a}\right\} + \left[\sqrt{p_{b}}-\left(\frac{p_{0}}{\sqrt{p_{b}}}-2\sqrt{p_{0}}\right) \right]\mathbbm{1}\left\{p_{t} \geq p_{b}\right\}\\
&  = - \frac{1}{\sqrt{p_{b}}} p_{t} + 2u_{1/2}(p_{t})+ \left(\frac{p_{0}}{\sqrt{p_{b}}}-2\sqrt{p_{0}}\right)  - \frac{1}{\sqrt{p_{b}}} \max\left\{ p_{a}-p_{t}, 0\right\} + \frac{1}{\sqrt{p_{b}}} \max\left\{ p_{t}-p_{b}, 0\right\} \\
& + \left[ \frac{p_{a}-p_{0}}{\sqrt{p_{b}}}  + 2\sqrt{p_{0}}  \right]\mathbbm{1}\left\{p_{t}\leq p_{a}\right\} + \left[\frac{p_{b}-p_{0}}{\sqrt{p_{b}}}+2\sqrt{p_{0}} \right]\mathbbm{1}\left\{p_{t} \geq p_{b}\right\}.
\end{split}
\end{equation}

Next we evaluate the put side for $p_{t}\leq p_{a}$:
\begin{equation} \label{eq:ilf7} 
\begin{split}
Put(p_{t}) & = p_{t} \left(\frac{1}{\sqrt{p_{a}}} - \frac{1}{\sqrt{p_{b}}}\right) + \frac{p_{0}}{\sqrt{p_{b}}}  - 2\sqrt{p_{0}} + \sqrt{p_{a}} \\
& = \left[ \frac{\sqrt{p_{b}}-\sqrt{p_{a}}}{\sqrt{p_{a}}\sqrt{p_{b}}}\right] \left(   p_{t} \pm p_{a} \right) + \left(\frac{p_{0}}{\sqrt{p_{b}}}  - 2\sqrt{p_{0}} + \sqrt{p_{a}} \right)\\
& = \left[ \frac{\sqrt{p_{b}}-\sqrt{p_{a}}}{\sqrt{p_{a}}\sqrt{p_{b}}}\right] \left( -\left\{p_{a}-p_{t} \right\} + p_{a}  \right)+  \left(\frac{p_{0}}{\sqrt{p_{b}}}  - 2\sqrt{p_{0}} + \sqrt{p_{a}} \right)  \\
& = -\left[ \frac{\sqrt{p_{b}}-\sqrt{p_{a}}}{\sqrt{p_{a}}\sqrt{p_{b}}}\right] \left\{p_{a}-p_{t}\right\} +  \left[ \frac{\sqrt{p_{b}}-\sqrt{p_{a}}}{\sqrt{p_{a}}\sqrt{p_{b}}}\right]  p_{a} +  \left(\frac{p_{0}}{\sqrt{p_{b}}}  - 2\sqrt{p_{0}} + \sqrt{p_{a}} \right)\\
& = -\left[ \frac{\sqrt{p_{b}}-\sqrt{p_{a}}}{\sqrt{p_{a}}\sqrt{p_{b}}}\right]  \left\{p_{a}-p_{t}\right\} +   \frac{\sqrt{p_{b}}\sqrt{p_{a}}-p_{a}}{\sqrt{p_{b}}} + \frac{p_{0} - 2\sqrt{p_{0}}\sqrt{p_{b}}+ \sqrt{p_{b}}\sqrt{p_{a}}}{\sqrt{p_{b}}}\\
& = -\left[ \frac{\sqrt{p_{b}}-\sqrt{p_{a}}}{\sqrt{p_{a}}\sqrt{p_{b}}}\right] \left\{p_{a}-p_{t} \right\} + \frac{ (p_{0}-p_{a}) - 2\sqrt{p_{0}}\sqrt{p_{b}}+ 2\sqrt{p_{b}}\sqrt{p_{a}}}{\sqrt{p_{b}}}\\
& = -\left[ \frac{\sqrt{p_{b}}-\sqrt{p_{a}}}{\sqrt{p_{a}}\sqrt{p_{b}}}\right]\left\{p_{a}-p_{t}\right\} + \frac{p_{0}-p_{a}}{\sqrt{p_{b}}} - 2(\sqrt{p_{0}}-\sqrt{p_{a}}).
\end{split}
\end{equation}

We further extend the last expression on $p_{t}\in(0, +\infty)$ 
\begin{equation} \label{eq:ilf8} 
\begin{split}
Put(p_{t}) & = -\left[ \frac{\sqrt{p_{b}}-\sqrt{p_{a}}}{\sqrt{p_{a}}\sqrt{p_{b}}}\right] \max \left\{p_{a}-p_{t}, 0 \right\} + \left[\frac{p_{0}-p_{a}}{\sqrt{p_{b}}} - 2(\sqrt{p_{0}}-\sqrt{p_{a}})\right] \mathbbm{1}\left\{p_{t}\leq p_{a}\right\}.
\end{split}
\end{equation}

Next we extend the call side for $p_{t}\in(0, +\infty)$ as follows
\begin{equation} \label{eq:ilf9} 
\begin{split}
 Call(p_{t}) &= \left[\sqrt{p_{b}}  + \frac{p_{0}}{\sqrt{p_{b}}}  - 2\sqrt{p_{0}}  \right] \mathbbm{1}\left\{p_{t} \geq p_{b}\right\}.
\end{split}
\end{equation}

Finally we sum up the three parts on $p_{t}\in(0, +\infty)$ as
\begin{equation} \label{eq:ilf9a} 
\begin{split}
 & \mathrm{Range}(p_{t})  + Put(p_{t}) + Call(p_{t}) \\
& = - \frac{1}{\sqrt{p_{b}}} p_{t} + 2u_{1/2}(p_{t})+ \left(\frac{p_{0}}{\sqrt{p_{b}}}-2\sqrt{p_{0}}\right)  - \frac{1}{\sqrt{p_{b}}} \max\left\{ p_{a}-p_{t}, 0\right\} + \frac{1}{\sqrt{p_{b}}} \max\left\{ p_{t}-p_{b}, 0\right\} \\
& + \left[ \frac{p_{a}-p_{0}}{\sqrt{p_{b}}}  + 2\sqrt{p_{0}}  \right]\mathbbm{1}\left\{p_{t}\leq p_{a}\right\} + \left[\frac{p_{b}-p_{0}}{\sqrt{p_{b}}}+2\sqrt{p_{0}} \right]\mathbbm{1}\left\{p_{t} \geq p_{b}\right\}\\
&  -\left[ \frac{\sqrt{p_{b}}-\sqrt{p_{a}}}{\sqrt{p_{a}}\sqrt{p_{b}}}\right] \max \left\{p_{a}-p_{t}, 0 \right\} + \left[\frac{p_{0}-p_{a}}{\sqrt{p_{b}}} - 2(\sqrt{p_{0}}-\sqrt{p_{a}})\right] \mathbbm{1}\left\{p_{t}\leq p_{a}\right\}\\
& + \left[ \sqrt{p_{b}}  + \frac{p_{0}}{\sqrt{p_{b}}}  - 2\sqrt{p_{0}}  \right] \mathbbm{1}\left\{p_{t} \geq p_{b}\right\}\\
& =   - \frac{1}{\sqrt{p_{b}}} p_{t} + 2u_{1/2}(p_{t})+ \left(\frac{p_{0}}{\sqrt{p_{b}}}-2\sqrt{p_{0}}\right)  \\
& -\left[ \frac{1}{\sqrt{p_{b}}} + \frac{\sqrt{p_{b}}-\sqrt{p_{a}}}{\sqrt{p_{a}}\sqrt{p_{b}}} \right]\max\left\{ p_{a}-p_{t}, 0\right\}  + \frac{1}{\sqrt{p_{b}}} \max\left\{ p_{t}-p_{b}, 0\right\} \\
& + \left[ \frac{p_{a}-p_{0}}{\sqrt{p_{b}}}  + 2\sqrt{p_{0}} + \frac{p_{0}-p_{a}}{\sqrt{p_{b}}} - 2(\sqrt{p_{0}}-\sqrt{p_{a}}) \right]\mathbbm{1}\left\{p_{t}\leq p_{a}\right\} \\
& + \left[\frac{p_{b}-p_{0}}{\sqrt{p_{b}}}+2\sqrt{p_{0}}+\sqrt{p_{b}}  + \frac{p_{0}}{\sqrt{p_{b}}}  - 2\sqrt{p_{0}}  \right]\mathbbm{1}\left\{p_{t} \geq p_{b}\right\}\\
& =  - \frac{1}{\sqrt{p_{b}}} p_{t} + 2u_{1/2}(p_{t})+ \left(\frac{p_{0}}{\sqrt{p_{b}}}-2\sqrt{p_{0}}\right) \\
& -\frac{1}{\sqrt{p_{a}}} \max\left\{ p_{a}-p_{t}, 0\right\}  + \frac{1}{\sqrt{p_{b}}} \max\left\{ p_{t}-p_{b}, 0\right\}  + 2\sqrt{p_{a}} \mathbbm{1}\left\{p_{t}\leq p_{a}\right\}  + 2\sqrt{p_{b}}\mathbbm{1}\left\{p_{t} \geq p_{b}\right\}.
\end{split}
\end{equation}

The last expression equals $u_{0}(p_{t}) + 2u_{1/2}(p_{t}) + u_{1}(p_{t})$ with the functions defined in Eq. \eqref{eq:il_dec_funded2}. Dividing by $\bar{V}_{0}$ gives Eq. \eqref{eq:il_dec_funded1}.

\end{proof}

\subsection{Proof of Proposition \protect\ref{prop:il_dec_borrowed1}}

\label{ap_prop:il_dec_borrowed1} 
\begin{proof}

Using P\&L of borrowed LP in Eq. \eqref{eq:pnl_borrowed_V3} we obtain the following.

\begin{equation} \label{eq:mdp1} 
\mathrm{P\&L\ borrowed}^{(y)}(p_{t}) =  \begin{cases}
  -\sqrt{p_{0}} \left(\sqrt{\frac{p_{t}}{p_{0}}}-1\right)^{2} 
\quad & \, p_{t} \in (p_{a}, p_{b}) \\
 p_{t} \left(\frac{1}{\sqrt{p_{a}}} - \frac{1}{\sqrt{p_{0}}}\right) -\left(\sqrt{p_{0}} - \sqrt{p_{a}}\right)  \quad & p_{t} \leq p_{a} \\
 \left(\sqrt{p_{b}} - \sqrt{p_{0}}\right) -p_{t} \left(\frac{1}{\sqrt{p_{0}}} - \frac{1}{\sqrt{p_{b}}}\right) \quad & p_{t}  \geq  p_{b}.
\end{cases}
\end{equation}
where $L=1$.

We decompose the negative of the P\&L, $-\mathrm{P\&L\ borrowed}^{(y)}(p_{t})$, into three parts.

The range part we evaluate for $p_{t} \in (p_{a}, p_{b}) $ as follows
\begin{equation} \label{eq:rp1} 
\begin{split}
\mathrm{Range}(p_{t}) & =  \sqrt{p_{0}} \left(\frac{p_{t}}{p_{0}} -2\sqrt{\frac{p_{t}}{p_{0}}}+1\right)\\
& \equiv    \left(\frac{1}{\sqrt{p_{0}}}u_{1}(p_{t}) -  2u_{1/2}(p_{t})+\sqrt{p_{0}} u_{0}(p_{t}) \right).
\end{split}
\end{equation}
It is clear that the first term can be decomposed for $p_{t}\in(0, \infty)$ as a collar option position :
\begin{equation} \label{eq:rp1a} 
\begin{split}
u_{1}(p_{t}) & =  p_{t} + \max\left\{ p_{a}-p_{t}, 0\right\} - \max\left\{ p_{t}-p_{b}, 0\right\} - p_{a}\mathbbm{1}\left\{p_{t}\leq p_{a}\right\} - p_{b}\mathbbm{1}\left\{p_{t} \geq p_{b}\right\}.
\end{split}
\end{equation}
where $\mathbbm{1}\{x\}$ is the indicator function. 

We evaluate $u_{0}(p_{t})$ term is as follows:
\begin{equation} \label{eq:rp2} 
\begin{split}
u_{0}(p_{t}) & =  1 - \mathbbm{1}\left\{p_{t} \geq p_{b}\right\} - \mathbbm{1}\left\{p_{t}\leq p_{a}\right\}.
\end{split}
\end{equation}

Then $u_{1/2}(p_{t})$ is the only non-linear claim which needs model valuation
\begin{equation} \label{eq:rp3} 
\begin{split}
u_{1/2}(p_{t}) & =  \sqrt{p_{t}} \mathbbm{1}\left\{ p_{a} < p_{t} < p_{b}\right\} 
\end{split}
\end{equation}

We evaluate the put side for $p_{t} \leq p_{a}$ as follows
\begin{equation} \label{eq:mdp2} 
\begin{split}
Put(p_{t}) & =  \left(\sqrt{p_{0}} - \sqrt{p_{a}}\right) -  p_{t} \left(\frac{\sqrt{p_{0}}-\sqrt{p_{a}}}{\sqrt{p_{0}p_{a}}} \right)\\
& = \left(\frac{\sqrt{p_{0}}-\sqrt{p_{a}}}{\sqrt{p_{0}p_{a}}} \right)  \left[\pm p_{a} + \sqrt{p_{0}p_{a}} -  p_{t} \right]\\
& = \left(\frac{\sqrt{p_{0}}-\sqrt{p_{a}}}{\sqrt{p_{0}p_{a}}} \right)  \left[ \left(p_{a}-p_{t} \right)  + \left(\sqrt{p_{0}p_{a}} -  p_{a}\right) \right]\\
& = \frac{\sqrt{p_{0}}-\sqrt{p_{a}}}{\sqrt{p_{0}p_{a}}} \left[ p_{a}-p_{t} \right]  + \frac{\sqrt{p_{0}}-\sqrt{p_{a}}}{\sqrt{p_{0}p_{a}}}\left(\sqrt{p_{0}p_{a}} -  p_{a}\right)\\
& = \frac{\sqrt{p_{0}}-\sqrt{p_{a}}}{\sqrt{p_{0}p_{a}}} \left[ p_{a}-p_{t} \right]  + \frac{\left(\sqrt{p_{0}}-\sqrt{p_{a}}\right)^{2}}{\sqrt{p_{0}}}.
\end{split}
\end{equation}
As a result, the pay-off on the put side can be written for $p_{t}\in(0, \infty)$ as follows
\begin{equation} \label{eq:mdp3} 
\begin{split}
Put(p_{t}) & = \frac{\sqrt{p_{0}}-\sqrt{p_{a}}}{\sqrt{p_{0}p_{a}}} \max\left\{ p_{a}-p_{t}, 0\right\}  + \frac{\left(\sqrt{p_{0}}-\sqrt{p_{a}}\right)^{2}}{\sqrt{p_{0}}} \mathbbm{1}\left\{p_{t}\leq p_{a}\right\}.
\end{split}
\end{equation}

Second, we evaluate the call side for $p_{t}  \geq  p_{b}$ as follows
\begin{equation} \label{eq:mdp4} 
\begin{split}
Call(p_{t}) & = p_{t} \left(\frac{\sqrt{p_{b}}-\sqrt{p_{0}}}{\sqrt{p_{0}p_{b}}} \right) -  \left(\sqrt{p_{b}} - \sqrt{p_{0}}\right)\\
& =  \left(\frac{\sqrt{p_{b}}-\sqrt{p_{0}}}{\sqrt{p_{0}p_{b}}} \right)  \left[ p_{t} - \sqrt{p_{0}p_{b}} \pm p_{b} \right]\\
& =  \left(\frac{\sqrt{p_{b}}-\sqrt{p_{0}}}{\sqrt{p_{0}p_{b}}} \right)  \left[ p_{t}-p_{b}\right] + \left(\frac{\sqrt{p_{b}}-\sqrt{p_{0}}}{\sqrt{p_{0}p_{b}}} \right) \left[p_{b} - \sqrt{p_{0}p_{b}} \right]\\
& =  \left(\frac{\sqrt{p_{b}}-\sqrt{p_{0}}}{\sqrt{p_{0}p_{b}}} \right)  \left[ p_{t}-p_{b}\right] + \frac{\left(\sqrt{p_{b}}-\sqrt{p_{0}}\right)^{2}}{\sqrt{p_{0}}}.
\end{split}
\end{equation}

As a result, the pay-off on the call side can be written for $p_{t}\in(0, \infty)$ as follows
\begin{equation} \label{eq:mdp5} 
\begin{split}
Call(p_{t}) & = \left(\frac{\sqrt{p_{b}}-\sqrt{p_{0}}}{\sqrt{p_{0}p_{b}}} \right)  \max\left\{ p_{t}-p_{b}, 0 \right\} + \frac{\left(\sqrt{p_{b}}-\sqrt{p_{0}}\right)^{2}}{\sqrt{p_{0}}} \mathbbm{1}\left\{p_{t} \geq p_{b}\right\}.
\end{split}
\end{equation}

Summing all together:
\begin{equation} \label{eq:rp10} 
\begin{split}
& \mathrm{Range}(p_{t}) + Put(p_{t}) + Call(p_{t}) \\
& =  -  2u_{1/2}(p_{t}) + \frac{1}{\sqrt{p_{0}}} \left[p_{t} + \max\left\{ p_{a}-p_{t}, 0\right\} - \max\left\{ p_{t}-p_{b}, 0\right\} - p_{a}\mathbbm{1}\left\{p_{t}\leq p_{a}\right\} - p_{b}\mathbbm{1}\left\{p_{t} \geq p_{b}\right\}\right] \\
& + \sqrt{p_{0}}\left[1 - \mathbbm{1}\left\{p_{t} \geq p_{b}\right\} - \mathbbm{1}\left\{p_{t}\leq p_{a}\right\} \right]\\
& + \frac{\sqrt{p_{0}}-\sqrt{p_{a}}}{\sqrt{p_{0}p_{a}}} \max\left\{ p_{a}-p_{t}, 0\right\}  + \frac{\left(\sqrt{p_{0}}-\sqrt{p_{a}}\right)^{2}}{\sqrt{p_{0}}} \mathbbm{1}\left\{p_{t}\leq p_{a}\right\}\\
& + \left(\frac{\sqrt{p_{b}}-\sqrt{p_{0}}}{\sqrt{p_{0}p_{b}}} \right)  \max\left\{ p_{t}-p_{b}, 0 \right\} + \frac{\left(\sqrt{p_{b}}-\sqrt{p_{0}}\right)^{2}}{\sqrt{p_{0}}} \mathbbm{1}\left\{p_{t} \geq p_{b}\right\}\\
& =  -  2u_{1/2}(p_{t}) + \frac{1}{\sqrt{p_{0}}} p_{t}   + \sqrt{p_{0}} \\
& + \left[ \frac{\sqrt{p_{0}}-\sqrt{p_{a}}}{\sqrt{p_{0}p_{a}}} +\frac{1}{\sqrt{p_{0}}}  \right]\max\left\{ p_{a}-p_{t}, 0\right\} \\
& + \left[\frac{\sqrt{p_{b}}-\sqrt{p_{0}}}{\sqrt{p_{0}p_{b}}}-\frac{1}{\sqrt{p_{0}}}  \right]\max\left\{ p_{t}-p_{b}, 0 \right\} \\
& +  \left[ \frac{\left(\sqrt{p_{0}}-\sqrt{p_{a}}\right)^{2}}{\sqrt{p_{0}}} - \frac{p_{a}}{\sqrt{p_{0}}} -\sqrt{p_{0}}\right] \mathbbm{1}\left\{p_{t}\leq p_{a}\right\}\\
& + \left[\frac{\left(\sqrt{p_{b}}-\sqrt{p_{0}}\right)^{2}}{\sqrt{p_{0}}} -\frac{p_{b}}{\sqrt{p_{0}}} -\sqrt{p_{0}}\right] \mathbbm{1}\left\{p_{t} \geq p_{b}\right\}\\
&  =  -  2u_{1/2}(p_{t}) + \frac{1}{\sqrt{p_{0}}} p_{t}   + \sqrt{p_{0}} \\
& + \frac{1}{\sqrt{p_{a}}}\max\left\{ p_{a}-p_{t}, 0\right\} \\
& -\frac{1}{\sqrt{p_{b}}}  \max\left\{ p_{t}-p_{b}, 0 \right\} \\
& -2\sqrt{p_{a}} \mathbbm{1}\left\{p_{t}\leq p_{a}\right\}\\
& -2\sqrt{p_{b}}\mathbbm{1}\left\{p_{t} \geq p_{b}\right\}.
\end{split}
\end{equation}

The final result follows by collecting pay-offs of vanilla puts and calls and digitals. The last expression equals $-\left[u_{0}(p_{t}) + 2u_{1/2}(p_{t}) + u_{1}(p_{t})\right]$ with the functions defined in Eq. \eqref{eq:il_dec_borrowed2}. Dividing by $-\bar{V}_{0}$ gives Eq. \eqref{eq:il_dec_borrowed1}.

\end{proof}

\subsection{Proof to Proposition \protect\ref{prop:replication_lp}}

\label{sc:replication_lp}

\begin{proof}

For the put side, it is clear that
\begin{equation} \label{eq:ph6} 
\begin{split}
\Pi(K) = \sum^{N}_{n=1}w_{n}P_{n}(K) =  \sum^{N}_{n=1}w_{n} \max \left(K_{n}-K, 0 \right).
\end{split}
\end{equation}

We define the first-order derivatives at discrete strike points as follows:
\begin{equation} \label{eq:ph6a} 
\begin{split}
& \delta IL(K_{n}) = \frac{IL(K_{n})-IL(K_{n-1})}{K_{n}-K_{n-1}}\\
& \delta\Pi(K_{n} ) = \frac{\Pi(K_{n})-\Pi(K_{n-1})}{K_{n}-K_{n-1}}.
\end{split}
\end{equation}

In particular for $K_{n} \in \mathcal{K}$:
\begin{equation} \label{eq:ph7} 
\begin{split}
& \Pi(K_{n}) = \sum^{N}_{n' \geq n}w_{n'}P_{n'}(K_{n}) =  \sum^{N}_{n' \geq n}w_{n'} \left(K_{n'}-K_{n} \right)\\
& \Pi(K_{n-1}) = \sum^{N}_{n' \geq n-1}w_{n'}\left(K_{n'}-K_{n-1}\right)\\
&\hphantom{\Pi(K_{n-1})} = \sum^{N}_{n' \geq n}w_{n'}\left(K_{n'}-K_{n-1}\right).
\end{split}
\end{equation}

As a result:
\begin{equation} \label{eq:ph8} 
\begin{split}
 \delta\Pi(K_{n} ) & = \frac{1}{{K_{n}-K_{n-1}}} \left( \Pi(K_{n})-\Pi(K_{n-1})\right)\\
& = \frac{1}{{K_{n}-K_{n-1}}} \left( \sum^{N}_{n' \geq n}w_{n'} \left(K_{n'}-K_{n} \right) -\sum^{N}_{n' \geq n}w_{n'} \left(K_{n'}-K_{n-1} \right)\right)\\
& = - \sum^{N}_{n' \geq n}w_{n'},
\end{split}
\end{equation}
and 
\begin{equation} \label{eq:ph9} 
\begin{split}
& \delta \Pi(K_{n})- \delta \Pi(K_{n-1}) = w_{n-1}.
\end{split}
\end{equation}

Let $IL_{\mathrm{lin}}$ denote the piecewise-linear interpolation of $IL$ on the put grid. For $p\in[K_{n-1},K_n]$,
\begin{equation} \label{eq:ph10}
IL_{\mathrm{lin}}(p)=IL(K_{n-1})+(p-K_{n-1})\delta IL(K_n).
\end{equation}
The proposed weights telescope to $\delta\Pi(K_n)=-\delta IL(K_n)$. Since $K_N=p_0$ and $\Pi(K_N)=IL(K_N)=0$, integration of these constant interval slopes gives
\begin{equation} \label{eq:ph11}
\Pi^{put}(p)=-IL_{\mathrm{lin}}(p),\qquad K_1\leq p\leq p_0.
\end{equation}
On the call grid the slope over $[K_{m-1},K_m]$ is $\sum_{j=1}^{m-1}w_j$. Setting it equal to $-\delta IL(K_m)$ gives $w_1=-\delta IL(K_2)$ and $w_m=-[\delta IL(K_{m+1})-\delta IL(K_m)]$ for $m=2,\ldots,M-1$. The common zero anchor $K_1=p_0$ then gives the same interpolation result on the call side. The endpoint choices $w_1^{put}=w_M^{call}=0$ preserve the outermost secant slopes outside the finite grids.

\end{proof}

\subsection{Proof of Proposition \protect\ref{prop:bsf}}

\label{ap_prop:bsf}

\begin{proof}

We use Eq. \eqref{eq:il_payoff_e1} as follows
\begin{equation} \label{eq:bs3} 
\begin{split}
& PV^{funded}(t, p_{t}) = - \frac{1}{\bar{V}_{0}}e^{-r(T-t)}\mathbb{E}^{\mathbb{Q}}\left[u^{funded}_{0}(x_{t}+x_{\tau}) + 2u_{1/2}(x_{t}+x_{\tau}) + u_{1}(x_{t}+x_{\tau})\right],\\
& PV^{borrowed}(t, p_{t})= - \frac{1}{\bar{V}_{0}}e^{-r(T-t)}\mathbb{E}^{\mathbb{Q}}\left[u^{borrowed}_{0}(x_{t}+x_{\tau}) + 2u_{1/2}(x_{t}+x_{\tau}) + u_{1}(x_{t}+x_{\tau})\right].
\end{split}
\end{equation}

We value each term separately. The linear part $U_{0}$ follows from the martingale condition in Eq. \eqref{eq:mgf1a}.


The square-root pay-off in Eq. (\ref{eq:mgf3}) is computed using:
\begin{equation} \label{eq:bs5}
\begin{split}
U_{1/2}(t,p_t)
&=e^{-r\tau}\mathbb E^{\mathbb Q}\left[
\sqrt{p_0e^{x_t+x_\tau}}\,
\mathbbm 1\{x_a<x_t+x_\tau<x_b\}\right]\\
&=e^{-r\tau}\sqrt{p_t}\exp\left\{\tfrac12(\mu-\tfrac12\sigma^2)\tau\right\}
\int_{-d_-(p_t,p_a)}^{-d_-(p_t,p_b)}
e^{\frac12\sigma\sqrt\tau z}\mathbf n(z)\,dz\\
&=e^{-r\tau}\sqrt{p_t}\exp\left\{\tfrac12(\mu-\tfrac12\sigma^2)\tau\right\}
\left[m(-d_-(p_t,p_b))-m(-d_-(p_t,p_a))\right]\\
&=e^{-r\tau}\sqrt{p_t}\exp\left\{\tfrac12\mu\tau-\tfrac18\sigma^2\tau\right\}\\
&\quad\times\left[\mathbf N\left(\frac{\ln(p_b/p_t)-(r-q)\tau}{\sigma\sqrt\tau}\right)
-\mathbf N\left(\frac{\ln(p_a/p_t)-(r-q)\tau}{\sigma\sqrt\tau}\right)\right].
\end{split}
\end{equation}
where $\mathbf{n}$ is normal pdf and 
\begin{equation}\label {eq:bs6}
\begin{split}
& m(x) = \exp\left\{\frac{1}{8}\sigma^{2} \tau \right\} \mathbf{N} (x-\frac{1}{2}\sigma \sqrt{\tau} ).
\end{split}
\end{equation}

The option part is computed using the option pricing formulas for BSM model.

\end{proof}

\subsection{Carr-Madan representation}

\label{sec:carrmadan} In \cite{Carr2001} a representation of an arbitrary,
twice-differentiable function (pay-off) in terms of put and call pay-offs was
given. Here we derive the same representation relaxing the smoothness
assumption, only requiring that its first derivative possesses the
generalized derivative everywhere.

To this end, assume that $f:\mathbb{R}\mapsto \mathbb{R}$, is such that $%
f^{\prime }$ has generalized derivative at every point, and fix arbitrary $%
S,\,F\geq0$. Then we have 
\begin{align}
f(S)&=f(F)+I_{\{S>F\}}\int_F^Sf^{\prime
}(u)\,du-I_{\{S<F\}}\int_S^Ff^{\prime }(u)\,du  \notag \\
&=f(F)+I_{\{S>F\}}\int_F^S\left[f^{\prime }(F)+\int_F^uf^{\prime \prime
}(v)\,dv\right]du- I_{\{S<F\}}\int_S^F\left[f^{\prime }(F)-\int_u^Ff^{\prime
\prime }(v)\,dv\right]du  \notag \\
&=f(F)+f^{\prime }(F)(S-F)+I_{\{S>F\}}\int_F^S\int_F^uf^{\prime \prime
}(v)\,dv\,du +I_{\{S<F\}}\int_S^F\int_u^F f^{\prime \prime }(v)\,dv\,du 
\notag \\
&=f(F)+f^{\prime }(F)(S-F)+\int_F^S f^{\prime \prime }(v)(S-v)^{+}dv+\int_0^F f^{\prime
\prime }(v)(v-S)^{+}dv,  \label{e:rs}
\end{align}
where in the last step we used Fubini's theorem. (Note that the upper limit
in the first integral in \eqref{e:rs} can be replaced by infinity.)


\begin{thebibliography}{xx}

\bibitem[\protect\citeauthoryear{Adams \textit{et al.}}{2020}]{V2} Adams H., Zinsmeister N.,
and D. Robinson. (2020). Uniswap V2 Core, \emph{White paper} %
\url{https://uniswap.org/whitepaper.pdf}

\bibitem[\protect\citeauthoryear{Adams \textit{et al.}}{2021}]{V3} Adams H., Zinsmeister N.,
Salem M., Keefer R., and D. Robinson. (2021). Uniswap V3 Core, \emph{White
paper} \url{https://uniswap.org/whitepaper-v3.pdf}


\bibitem[\protect\citeauthoryear{Adams \textit{et al.}}{2023}]{V4} Adams H., Salem M., Zinsmeister N., Reynolds S., Adams, A., Pote, W., Toda, M., Henshaw, A., Williams, E., and D. Robinson. (2023). Uniswap V4 Core, \emph{White
paper} \url{https://github.com/Uniswap/v4-core/blob/main/docs/whitepaper-v4.pdf}


\bibitem[\protect\citeauthoryear{Angeris \textit{et al.}}{2019}]{Angeris2019} %
Angeris, G., Kao, H.T., Chiang, R., Noyes, C., Chitra, T., (2019). An
analysis of Uniswap markets \emph{Cryptoeconomic Systems Journal}


\bibitem[\protect\citeauthoryear{Alexander-Imeraj}{2023}]{Alexander2023a} %
Alexander, C., Imeraj, A. \harvardyearleft 2023\harvardyearright ,
Delta hedging bitcoin options with a smile. \emph{Quantitative Finance} \textbf{23}%
(5), ~799--817.


\bibitem[\protect\citeauthoryear{Alexander \textit{et al.}}{2023}]{Alexander2023} %
Alexander, C., Chen, D., Imeraj, A. \harvardyearleft 2023\harvardyearright ,
Crypto quanto and inverse options, \emph{Mathematical Finance} \textbf{33}%
(4), ~1005--1043.


\bibitem[\protect\citeauthoryear{Alexander \textit{et al.}}{2024}]{Alexander2024} %
Alexander, C., Chen, X., Deng, J., Fu, Q. \harvardyearleft 2024\harvardyearright ,
Impacts of Protocol Updates on Uniswap Decentralised Crypto Exchanges, \emph{Working paper}, \url{https://ssrn.com/abstract=4495589} 




\bibitem[\protect\citeauthoryear{Bergault \textit{et al.}}{2023}]{Bergault2023} Bergault, P.,
Bertucci, L., Bouba, D. (2023). Automated market makers: mean-variance
analysis of LPs pay-offs and design of pricing functions. \emph{Digital
Finance}, 6, 225--247, \url{https://doi.org/10.1007/s42521-023-00101-0}


\bibitem[\protect\citeauthoryear{Br\"onnimann \textit{et al.}}{2024}]{Bronnimann2024} Br\"onnimann, W., Egloff, P., Krabichler, T. (2024). Automated market makers and their implications
for liquidity providers. \emph{Digital
Finance}, 6, 573--604 \url{https://doi.org/10.1007/s42521-024-00117-0}


\bibitem[\protect\citeauthoryear{Capponi-Jia}{2024}]{Capponi2024} Capponi A., and
Jia, R. (2024). Liquidity provision on blockchain-based decentralised
exchanges. \emph{Working paper}, \url{https://ssrn.com/abstract=3805095}

\bibitem[\protect\citeauthoryear{Cartea \textit{et al.}}{2024}]{Cartea2024} Cartea A., Drissi
F., and Monga, M. (2024). Predictable losses of liquidity provision in
constant function markets and concentrated liquidity markets. \emph{Applied
Mathematical Finance}, 30, 69--93

\bibitem[\protect\citeauthoryear{Cartea \textit{et al.}}{2024}]{Cartea2024a} Cartea A., Drissi
F., and Monga, M. (2024). Decentralised Finance and Automated Market Making:
Predictable Loss and Optimal Liquidity Provision. \emph{Working paper}, %
\url{https://arxiv.org/abs/2309.08431}

\bibitem[\protect\citeauthoryear{Cartea \textit{et al.}}{2023}]{Cartea2023} Cartea A., Drissi
F., Sanchez-Betancourt L., Siska D., and Szpruch, L. (2023). Automated Market Makers Designs Beyond Constant Functions, \emph{Working paper}, %
\url{https://ssrn.com/abstract=4459177}


\bibitem[\protect\citeauthoryear{Carr-Madan}{2001}]{Carr2001} Carr P., Madan. D. 
(2001). Towards a Theory of Volatility Trading, In: E. Jouini, J. Cvitanic,
M. Musiela, eds. \emph{Handbooks in Mathematical Finance: Option Pricing,
Interest Rates and Risk Management}, Cambridge University Press 458--476


\bibitem[\protect\citeauthoryear{Deng \textit{et al.}}{2023}]{Deng2023} Deng J., Zong H., and
Y. Wang. (2023). Static replication of impermanent loss for concentrated
liquidity provision in decentralised markets, \emph{Operations Research
Letters} 51(3) 206--211

\bibitem[\protect\citeauthoryear{Echenim \textit{et al.}}{2023}]{Echenim2023} Echenim M.,
Gobet E., and A.C. Maurice. (2023). Thorough mathematical modelling and
analysis of Uniswap V3, \emph{Working paper} %
\url{https://hal.science/hal-04214315V2}

\bibitem[\protect\citeauthoryear{Elsts} {2021}]{Elsts2021} Elsts A. (2021). Liquidity math in Uniswap v3, \emph{SSRN preprint} %
\url{https://papers.ssrn.com/sol3/papers.cfm?abstract_id=4575232}

\bibitem[\protect\citeauthoryear{Fukasawa \textit{et al.}}{2023}]{Fukasawa2023} Fukasawa M.,
Maire B., and M. Wunsch. (2023). Weighted variance swaps hedge against
impermanent loss, \emph{Quantitative Finance} 23(6) 901--911

\bibitem[\protect\citeauthoryear{Heimbach \textit{et al.}}{2022}]{Heimbach2022} Heimbach L.,
Schertenleib E., and R. Wattenhofer. (2022). Risks and returns of Uniswap v3
liquidity providers, \emph{Working paper} %
\url{https://arxiv.org/abs/2205.08904}

\bibitem[\protect\citeauthoryear{Lehar-Parlour}{2024}]{Lehar2024} Lehar A., and
Parlour C.A. (2024). Decentralised Exchange: The Uniswap Automated Market
Maker, \emph{Journal of Finance, forthcoming} %
\url{https://ssrn.com/abstract=3905316}

\bibitem[\protect\citeauthoryear{Lewis}{2000}]{Lewis2000} Lewis, A.~L. (2000), \emph{%
Option Valuation under Stochastic Volatility}, Finance Press.

\bibitem[\protect\citeauthoryear{Li \textit{et al.}}{2023}]{Li2023} Li T., Naik S.,
Papanicolaou A., and L. Schoenleber. (2023). Yield farming for liquidity
provision, \emph{Working paper} \url{https://ssrn.com/abstract=4422213}

\bibitem[\protect\citeauthoryear{Li \textit{et al.}}{2024}]{Li2024} Li T., Naik S.,
Papanicolaou A., and L. Schoenleber. (2024). Implied Impermanent Loss: A
Cross-Sectional Analysis of Decentralised Liquidity Pools, \emph{Working
paper} \url{https://ssrn.com/abstract=4811111}

\bibitem[\protect\citeauthoryear{Lipton}{2001}]{Lipton2001} Lipton, A. (2001).
Mathematical Methods for Foreign Exchange: A Financial Engineer's Approach.
World Scientific

\bibitem[\protect\citeauthoryear{Lipton}{2018}]{Lipton2018} Lipton, A. (2018).
Financial Engineering: Selected Works of Alexander Lipton.
World Scientific

\bibitem[\protect\citeauthoryear{Lipton}{2024}]{Lipton2024} Lipton, A. (2024).
Hydrodynamics of Markets: Hidden Links Between Physics and Finance, \emph{%
Cambridge University Press}, \url{https://www.cambridge.org/core/elements/hydrodynamics-of-markets/1518AE3A15E7459C3847C946B6385D31}

\bibitem[\protect\citeauthoryear{Lipton-Hardjono}{2022}]{LiptonHardjono2022} Lipton,
A. and Hardjono T. (2022). Blockchain Intra- and Interoperability, in Babich
V, Birge J, Hilary G (eds), \emph{Springer, New York.}

\bibitem[\protect\citeauthoryear{Lipton-Sepp}{2008}]{LiptonSepp2008} Lipton, A. and
Sepp A. (2008). Stochastic volatility models and Kelvin waves.\emph{\ J. Phys. A: Math.
Theor.} \textbf{41}, 344012 (23pp)

\bibitem[\protect\citeauthoryear{Lipton-Sepp}{2022}]{LiptonSepp2022} Lipton, A. and
Sepp A. (2022). Automated market-making for fiat currencies, \emph{Risk
Magazine}, May

\bibitem[\protect\citeauthoryear{Lipton-Treccani}{2021}]{LiptonTreccani2021} %
Lipton A., and A. Treccani, (2021). Distributed Ledgers: Mathematics,
Technology, and Economics, \emph{World Scientific}, Singapore

\bibitem[\protect\citeauthoryear{Lucic-Sepp}{2024}]{LucicSepp2024} Lucic, V. and
Sepp, A. (2024), Valuation and Hedging of Cryptocurrency Inverse Options:
With Backtest Simulations using Deribit Options Data, \emph{Quantitative
Finance}, 2024, 24(7), 851--869 \url{https://ssrn.com/abstract=4606748}

\bibitem[\protect\citeauthoryear{Maire-Wunsch}{2024}]{MaireWunsch2024} Maire, B.
and Wunsch, M. (2024), Market Neutral Liquidity Provision, \emph{Preprint}

\bibitem[\protect\citeauthoryear{Matic \textit{et al.}}{2023}]{Matic2023} J. Matic, N. Packham, and W. K. H\"ardle. (2023), Hedging cryptocurrency options, \emph{Review of Derivatives Research} 26, 91--133

\bibitem[\protect\citeauthoryear{Milionis \textit{et al.}}{2022}]{Milionis2022} Milionis, J.,
Moallemi C. C., Roughgarden, T., Lee Zhang, A. (2022), Automated Market
Making and Loss-Versus-Rebalancing, \emph{Working paper} %
\url{https://api.semanticscholar.org/CorpusID:251554626}

\bibitem[\protect\citeauthoryear{Mohan}{2022}]{Mohan2022} Mohan, V. (2022),
Automated market makers and decentralised exchanges: a DeFi primer, \emph{%
Financial Innovation} 8(20) \url{https://doi.org/10.1186/s40854-021-00314-5}

\bibitem[\protect\citeauthoryear{Park}{2023}]{Park2023} Park, A. (2023), The
conceptual flaws of decentralised automated market making, \emph{Management
Science} 69(11) 6731--6751. \url{https://doi.org/10.1287/mnsc.2021.02802}

\bibitem[\protect\citeauthoryear{Sepp-Rakhmonov}{2023}]{SeppRakhmonoV2023} Sepp,
A. and Rakhmonov, P. (2023) Log-normal Stochastic Volatility Model with
Quadratic Drift, \emph{International Journal of Theoretical and Applied
Finance} 26(08), 2450003, \url{https://doi.org/10.1142/S0219024924500031}


\end{thebibliography}
\end{document}